\documentclass[a4paper,USenglish,cleveref, autoref, thm-restate]{lipics-v2021}

\usepackage{amsmath}
\usepackage{amssymb}
\usepackage{booktabs}
\usepackage{xspace}
\usepackage{graphicx}
\usepackage{subcaption}
\usepackage{todonotes}
\usepackage{xcolor}

\definecolor{colorRandom}{HTML}{7570B3}  %
\definecolor{colorSA}{HTML}{D95F02}      %
\definecolor{colorMILP}{HTML}{1B9E77}    %

\newcommand{\mRandom}{\textcolor{colorRandom}{\textbf{Random}}}
\newcommand{\mSA}{\textcolor{colorSA}{\textbf{SA}}}
\newcommand{\mMILP}{\textcolor{colorMILP}{\textbf{MILP}}}

\title{Minimum-Width Drawing of Trees with Sized Vertices}

\titlerunning{Minimum-Width Tree Drawings with Sized Vertices} %

\author{Markus Wallinger}{Technical University of Munich, Germany}{markus.wallinger@tum.de}{https://orcid.org/0000-0002-2191-4413}{}%

\author{Oscar Navarro}{Technical University of Munich, Germany}{oscar.navarro@tum.de}{https://orcid.org/0009-0003-7052-8345}{}%

\author{Stephen G. Kobourov}{Technical University of Munich, Germany}{stephen.kobourov@tum.de}{https://orcid.org/0000-0002-0477-2724}{}%

\authorrunning{M.\,Wallinger, O.\,Navarro, and S.\,Kobourov}

\Copyright{Markus Wallinger, Oscar Navarro, and Stephen Kobourov}



\category{Track 2, Long Paper} %

\relatedversion{} %

\supplement{Source code and benchmarks are available at \url{https://osf.io/dfkau}.}

\nolinenumbers %

\EventEditors{Debajyoti Mondal and Rahnuma Islam Nishat}
\EventNoEds{2}
\EventLongTitle{34th International Symposium on Graph Drawing and Network Visualization (GD 2026)}
\EventShortTitle{GD 2026}
\EventAcronym{GD}
\EventYear{2026}
\EventDate{August 19--21, 2026}
\EventLocation{St.\ Catharines, Ontario, Canada}
\EventLogo{}
\SeriesVolume{XXX}%
\ArticleNo{XX}%

\newcommand{\mwtdr}{\textsc{MWD}\xspace}

\begin{document}

\maketitle

\begin{figure}[h!]
  \centering
  \begin{subfigure}[b]{0.325\textwidth}
    \centering
    \includegraphics[width=\textwidth]{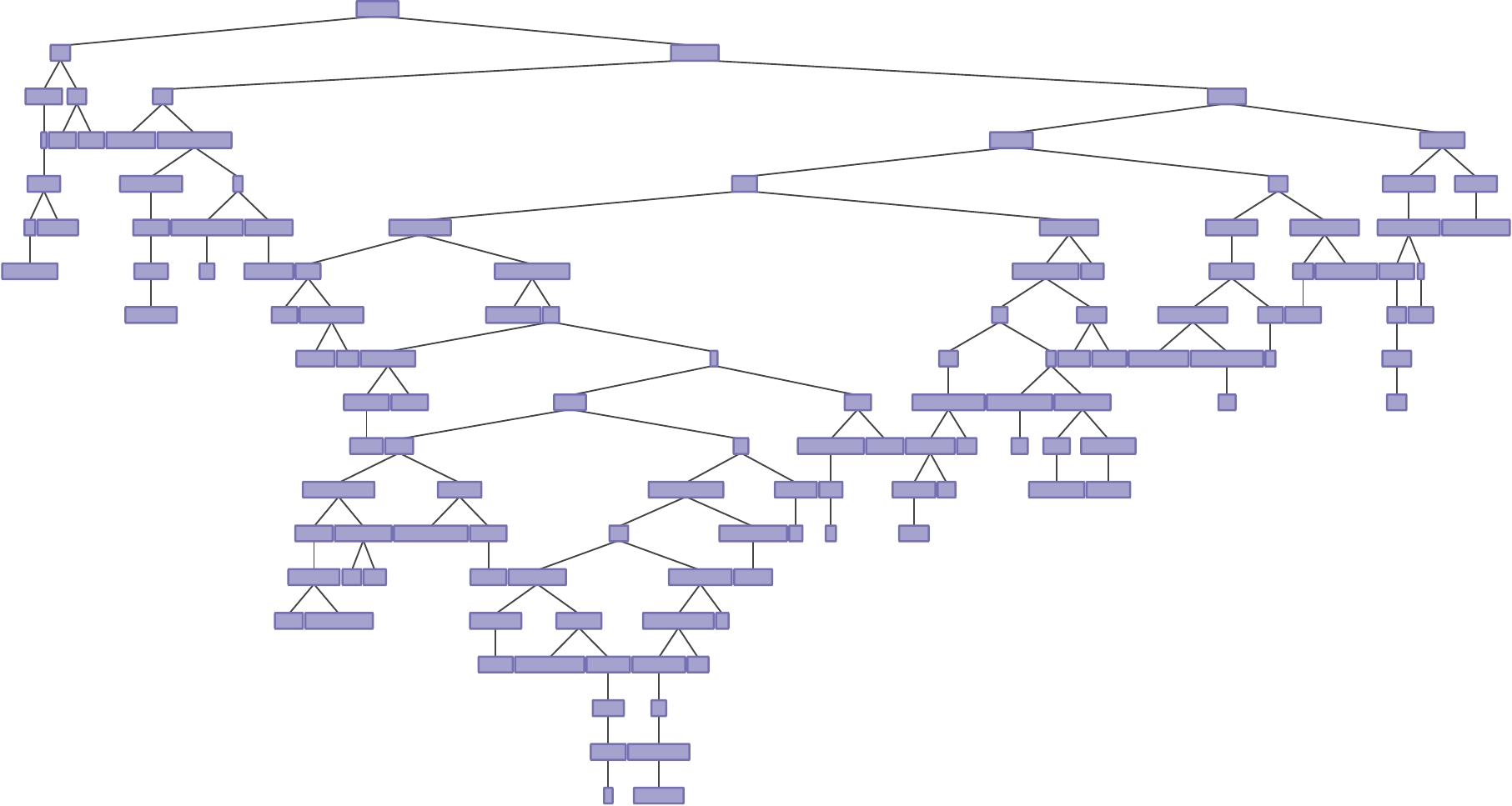}
    \label{fig:teaser-input}
  \end{subfigure}
  \hfill
  \begin{subfigure}[b]{0.325\textwidth}
    \centering
    \includegraphics[width=\textwidth]{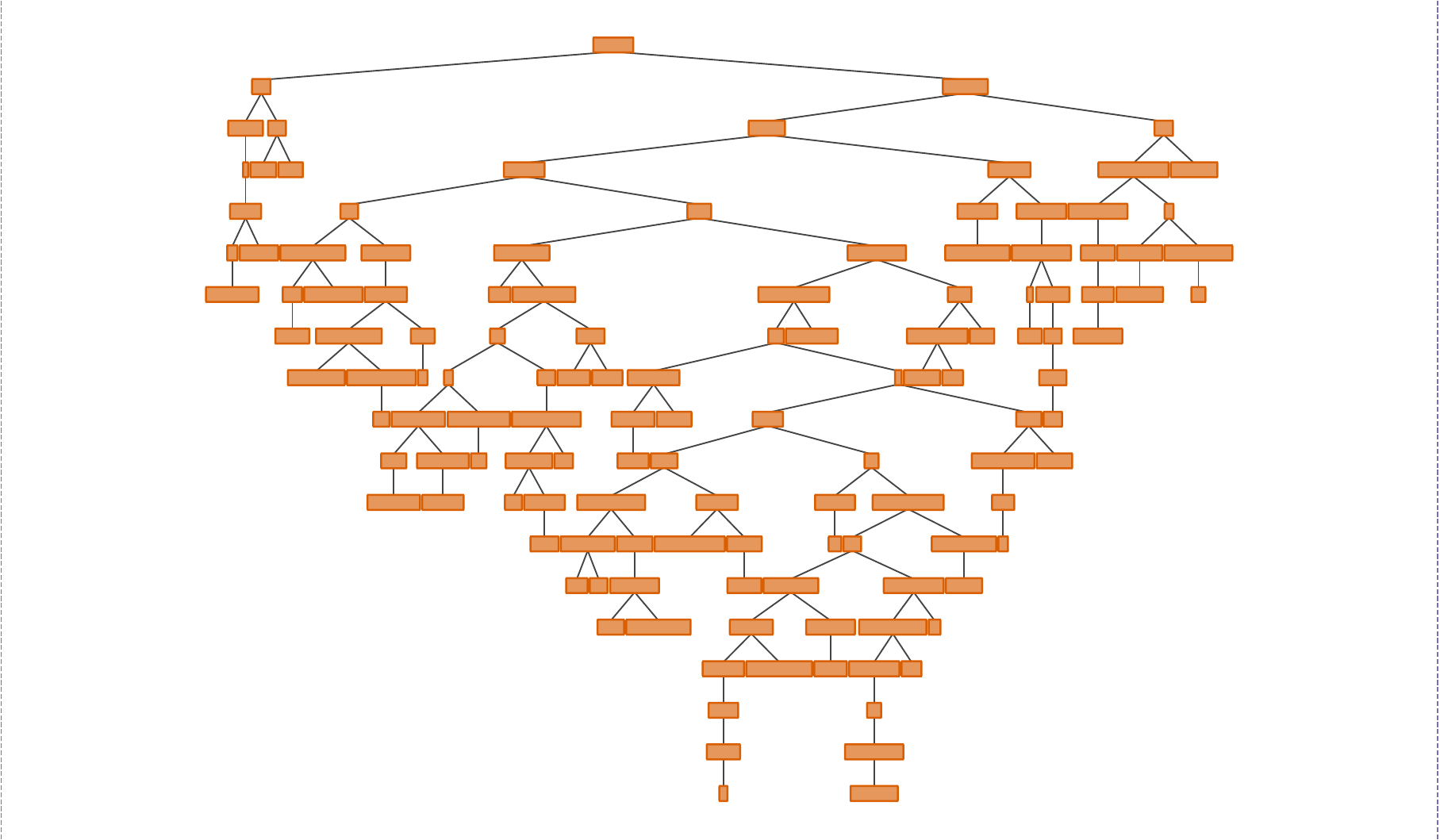}
    \label{fig:teaser-heuristic}
  \end{subfigure}
  \hfill
  \begin{subfigure}[b]{0.325\textwidth}
    \centering
    \includegraphics[width=\textwidth]{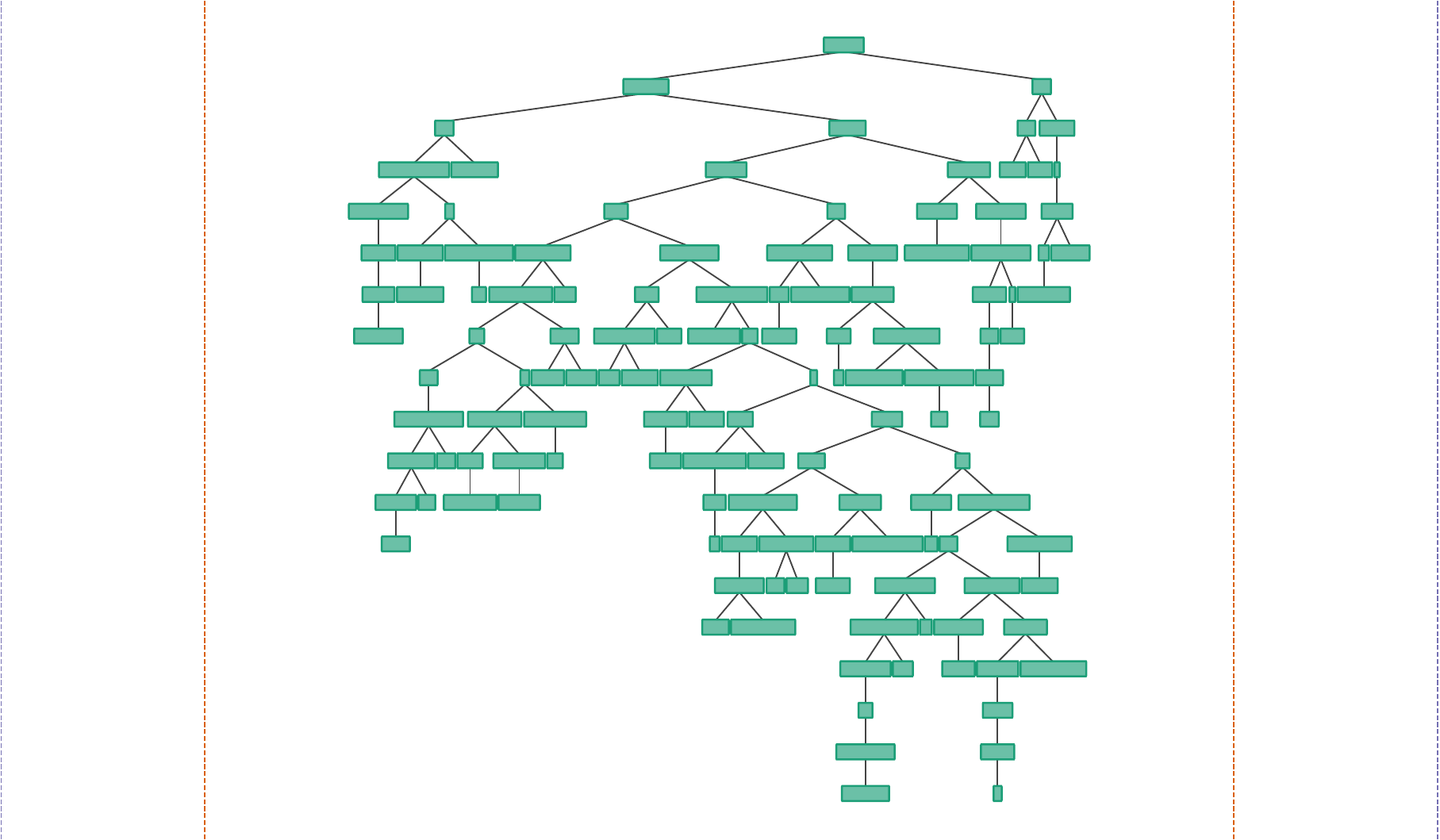}
    \label{fig:teaser-optimal}
  \end{subfigure}
  \caption{The same unordered tree drawn (left) in random order, (middle) under the order found by our simulated annealing heuristic, and (right) under a width-optimal order computed by the MILP. Reordering reduces the drawing width by $28.39\%$ and $48.53\%$, respectively (vertical bars).}
  \label{fig:teaser}
\end{figure}

\begin{abstract}
Trees arise in many applications and computing nice tree layouts is a classical problem in information visualization. 
In many practical settings, vertices need to be represented as rectangles with a given width and height rather than as points. 
When an order over the children of each vertex is given, polynomial-time algorithms are known that produce drawings adhering to various drawing conventions. However, in many applications, the order of children carries no semantic meaning, and choosing it well can significantly reduce the drawing's width.
In this paper, we study the problem \textsc{Min-Width Tree Drawing with Reordering} (\textsc{MWD}): given a rooted tree whose vertices have prescribed widths, find a sibling order at each internal vertex that minimizes the width of the resulting layered drawing. We show that the problem is \textsf{NP}-complete, even on binary trees with unit-width vertices. 
We present a mixed integer linear program that solves \textsc{MWD} exactly on moderately sized instances, and a heuristic that is fast and delivers good results in practice. 
We evaluate both approaches against a baseline on synthetic and real-world datasets, where reordering reduces drawing width by a median of $\approx20\%$ and by up to $\approx55\%$ on individual instances.
The heuristic computes its layouts in under a second and, when the MILP proves optimality, it stays within $25\%$ of the optimal width in three-quarters of all instances.

\subparagraph{Generative AI Declaration}
Claude Opus 4.7 was used throughout the research process to generate code and refine writing.
\end{abstract}

\section{Introduction}
\label{sec:intro}

Trees are among the most ubiquitous abstractions in computer science, and node-link diagrams are widely used to visualize them~\cite{Schulz11}.
In a \emph{tidy drawing} of a rooted tree, vertices at the same depth are vertically aligned, parents are centered above their children, and no two vertices overlap.
This convention was introduced by Wetherell and Shannon~\cite{WetherellS79} and is now one standard in tree visualization.
The classic Reingold--Tilford algorithm~\cite{ReingoldT81} produces tidy layered drawings of rooted trees with point as vertices in linear time.
More recently, van der Ploeg~\cite{Ploeg14} extended these results to vertices with inherent widths and heights, both in the layered convention and in a non-layered relaxation in which a child is placed immediately below its parent rather than at a vertically aligned depth.
A common feature of such tidy drawing algorithms is that they take the \emph{order} of children at each internal node as part of the input.
This is well-motivated for trees arising from parse or syntax trees, where the order is semantically meaningful.
In many other applications, however, the underlying tree is unordered: examples include taxonomies, file-system hierarchies, organization charts, class hierarchies in software engineering, and phylogenetic trees.
In such settings, the sibling order is a degree of freedom, and a poor choice can lead to much wider than necessary layouts.

The computational complexity of choosing a width-optimal sibling order has been studied in several drawing conventions~\cite{Hayashi98,MarriottS04,VialGJO19,Biedl22}.
None of this prior work addresses the layered tidy setting with vertices of prescribed widths, which is the gap we close.
\Cref{fig:teaser} shows the same unordered tree under three sibling orders, illustrating how reordering reduces width.

\textbf{Our contributions.}
We introduce and study the \textsc{Min-Width Tree Drawing with Reordering} (\mwtdr) problem: given a rooted tree $T$ with vertex widths, find a sibling order at each internal node so that the resulting tidy layered drawing has minimum width.
Concretely, we contribute:
\begin{enumerate}
    
    \item \textbf{An \textsf{NP}-completeness proof} (\cref{sec:hardness}) by reduction from Marriott--Stuckey's unordered layered binary tree layout~\cite{MarriottS04}.
    The reduction sets all vertex widths to one, showing that the hardness of \mwtdr comes from the order choice alone and is not an artifact of having sized vertices.

    \item \textbf{An exact mixed-integer linear program} (\cref{sec:milp}) that solves \mwtdr on moderately sized instances and serves as a ground-truth oracle in our experiments.

    \item \textbf{A heuristic} (\cref{sec:heuristic}) based on simulated annealing with van der Ploeg's drawing algorithm~\cite{Ploeg14} as a subroutine.

    \item \textbf{An experimental evaluation} (\cref{sec:experiments}) on synthetic benchmarks and real-world instances, comparing the heuristic against the ILP optimum and against a natural baseline.
\end{enumerate}

\section{Related Work}
\label{sec:related}

For broader context on tree-drawing conventions, see the surveys by Schulz~\cite{Schulz11} and Rusu~\cite{Rusu13}.

\textbf{Tidy drawings of trees.}
A drawing of a rooted tree is called \emph{tidy} if vertices at the same depth are vertically aligned in a layer, every parent is horizontally centered above its children, and no two vertices overlap.
This convention has been the dominant model for tree visualization since the linear-time algorithms of Knuth~\cite{Knuth68} and Wetherell and Shannon~\cite{WetherellS79}.
Reingold and Tilford~\cite{ReingoldT81} strengthened the aesthetic guarantees by requiring identical subtrees to be drawn identically.
Walker~\cite{Walker90} added support for unbounded degree and reflection symmetry. Walker's original algorithm runs in $O(n^2)$, but Buchheim et al.~\cite{BuchheimJL06} showed a linear time implementation.
All of these algorithms target trees with point-like vertices with real-valued coordinates, using a layered convention.
Van der Ploeg~\cite{Ploeg14} extended the linear-time guarantee to vertices with prescribed widths and heights and gave both layered and non-layered variants of the resulting algorithm.
None of these algorithms chooses the sibling order; instead, they assume it is part of the input.

\textbf{Width minimization on unordered trees.}
The complexity of width-optimal drawings of unordered trees has been studied across several drawing conventions.
Supowit and Reingold~\cite{SupowitR82} showed that drawing trees nicely is \textsf{NP}-hard when the horizontal coordinates must be integers, whereas the same problem with real-valued coordinates reduces to linear programming and is solvable in polynomial time.
Hayashi and Masuda~\cite{Hayashi98} obtained complexity results for minimum-width drawings of rooted unordered trees.
Marriott and Stuckey~\cite{MarriottS04} introduced the unordered layered tree drawing convention for binary rooted trees in which a parent's $x$-coordinate is the average of its children's $x$-coordinates and proved that finding a minimum-width drawing in this convention is \textsf{NP}-complete by reduction from \textsc{SAT}, notably without requiring integer coordinates.
Their reduction is the basis of our hardness proof in \cref{sec:hardness}.
More recently, Besa et al.~\cite{VialGJO19} established \textsf{NP}-hardness for minimum-width orthogonal upward drawings of phylogenetic trees with fixed edge lengths, gave a linear-time algorithm for the ordered case based on a horizontal-visibility constraint graph, and evaluated several heuristics on trees from \textsc{TreeBase}.
Their drawing convention (orthogonal, upward, with leaves allowed to lie above neighboring branches) differs from ours in that there is no centering constraint and vertices have only an inherent height (set by edge length) but not an inherent width.
Misue~\cite{Misue24} relaxes the same-layer horizontal-line constraint by 
folding sibling leaves onto multiple rows, optimizing for area adaptivity 
rather than width.
Klawitter and Zink~\cite{KlawitterZ23} study a related problem on tree 
drawings constrained to vertical columns, where the objective is to 
minimize edge crossings rather than width.

\textbf{Polynomial cases for unordered trees.}
The results in this paragraph concern grid drawings, where vertices are placed at integer coordinates, and the width is measured in grid columns.
For straight-line upward drawings with point-like vertices, Biedl~\cite{Biedl22} showed that the minimum width in the unordered model equals the rooted pathwidth and can be computed in linear time.
Eades et al.~\cite{EadesLL92} %
gave polynomial algorithms for minimum-width and minimum-area unordered HV-drawings (edges must be horizontal or vertical line segments) of binary trees.
Bachmaier and Matzeder~\cite{BachmaierM13} studied unordered tree drawings on $k$-grids.
Garg and Rusu~\cite{GargR04} gave a linear-area straight-line algorithm for unordered binary trees with user-controlled aspect ratio.
The common thread is that polynomial-time results exist precisely when the drawing convention removes the centering constraint that makes the variant studied by Marriott and Stuckey hard, or when vertices are points rather than rectangles with prescribed widths.

\textbf{Sized vertices in graph drawing.}
Di Battista et al.~\cite{BattistaDPP99} considered orthogonal and quasi-upward drawings with vertices of prescribed size, but for general graphs and without the unordered tree question.
Marriott et al.~\cite{MarriottSGPB11} discuss hi-trees, in which the parent--child relationship is visualized either by a link or by containment, and use the Reingold--Tilford pipeline as a subroutine.
We are not aware of prior work that combines (i) the layered tidy convention, (ii) sized vertices, and (iii) the freedom to choose the sibling order, which is the setting of this paper. Like the unordered layered convention discussed above, our model uses real-valued coordinates.

\section{Background and Notation}
\label{sec:background}

\begin{figure}[t]
  \centering
  \includegraphics[width=0.8\textwidth]{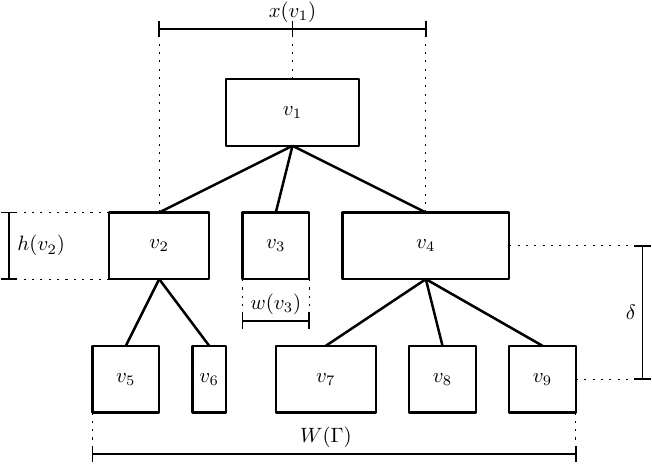}
  \caption{A tidy layered drawing of a sized tree, with bounding boxes of prescribed widths $w(v)$, layer distance $\delta$, and total drawing width $W(\Gamma)$. The inset shows the centering rule (A3): $v_1$ sits at the midpoint of the centers of its leftmost child $v_2$ and rightmost child $v_4$; the middle child $v_3$ does not influence the parent's position.}
  \label{fig:background}
\end{figure}

We illustrate notation in \Cref{fig:background}. Throughout the paper, $T = (V, E)$ denotes a finite rooted tree with vertex set $V$ and edge set $E$, with a designated root $r \in V$.
For every non-root $v \in V$, we write $\mathrm{parent}(v)$ for its parent and $\mathrm{children}(v)$ for the (possibly empty) set of its children.
We denote by $n = |V|$ the number of vertices and by $\Delta$ the maximum number of children of any vertex.
The \emph{depth} $d(v)$ of a vertex is the length of the unique path from $r$ to $v$, and the \emph{depth} of $T$ is $\max_v d(v)$.
We write $L_i = \{v \in V : d(v) = i\}$ for the set of vertices at depth $i$.

A \emph{sibling order} of $T$ is, for each internal vertex $v$, a permutation $\pi_v$ of $\mathrm{children}(v)$.
A pair $(T, \pi)$ in which a sibling order has been fixed at every internal vertex is an \emph{ordered tree}.

In our setting, every vertex $v \in V$ has a prescribed \emph{width} $w(v) \in \mathbb{R}_{>0}$ and \emph{height} $h(v) \in \mathbb{R}_{>0}$.
In the rest of this paper we assume that the height is a fixed constant for all vertices, e.g., corresponding to the font size used for labels.
We refer to the pair $(w, h)$ collectively as the \emph{sizes} of $T$ and call $T$ together with its sizes a \emph{sized tree}.

\subsection{Tidy layered drawings and aesthetic criteria}

We adopt the \emph{layered} drawing convention of Reingold and Tilford~\cite{ReingoldT81}.
A \emph{drawing} $\Gamma$ of a sized ordered tree $(T, \pi, w)$ assigns to each $v \in V$  horizontal $x(v) \in \mathbb{R}$ and vertical $y(v) \in \mathbb{R}$ coordinates representing the center of $v$'s axis-aligned bounding box.
The vertical position is determined entirely by depth, with a fixed layer distance $\delta>0$ .
Each vertex $v$ thus occupies the axis-aligned rectangle $B(v) = [x(v) - w(v)/2,\; x(v) + w(v)/2] \times [y(v) - h(v)/2,\; y(v) + h(v)/2]$
Every vertex at depth $i$ has the same vertical coordinate, and only the
horizontal coordinate is a degree of freedom.

\textbf{Enforced Aesthetics.}
A drawing is \emph{tidy} if it satisfies three aesthetic criteria, due to Reingold and Tilford~\cite{ReingoldT81} and naturally extending to sized vertices:

\begin{description}
  \item[(A1) No overlap.] $B(u) \cap B(v) = \emptyset$, for all $u, v \in V$ with $u \neq v$.
  \item[(A2) Sibling order.] For every internal vertex $v$ with children $c_1, c_2, \dots, c_k$ in the order given by $\pi_v$, we have $x(c_1) < x(c_2) < \cdots < x(c_k)$, so that the chosen $\pi$ constrains the geometric layout.
  \item[(A3) Centered parents.] For every internal vertex $v$ with leftmost child $c_L$ and rightmost child $c_R$ under $\pi_v$,
  \begin{equation}
  x(v) = \tfrac{1}{2}\bigl(x(c_L) + x(c_R)\bigr).
  \label{eq:centering}
  \end{equation}
  When $v$ has a single child $c$, this reduces to $x(v) = x(c)$.
\end{description}

For point-like vertices (the original Reingold--Tilford setting) and for vertices of equal width, rule A3 coincides with the alternative formulation that centers the parent's bounding box above the joint bounding box of its leftmost and rightmost children. For vertices of varying width, the midpoint-of-centers rule used here is the more compact convention, as it allows a parent to overhang the wider extreme child, so that adjacent parent subtrees can interleave their overhangs at upper layers, whereas the joint-bounding-box rule confines each parent within its children's joint extent and propagates extra width upward.

The \emph{width} of a drawing is the horizontal extent of the union of all bounding boxes:
\begin{equation}
W(\Gamma) = \max_{v \in V}(x(v) + w(v)/2) - \min_{v \in V}(x(v) - w(v)/2).
\label{eq:width}
\end{equation}

\textbf{Aesthetics we do not enforce.}
The classical Reingold--Tilford convention includes \emph{subtree isomorphism} (isomorphic subtrees drawn congruently) and \emph{reflection symmetry} (the drawing of the reflected tree is the mirror image of the original), neither of which we enforce.
Subtree isomorphism is automatically satisfied by recursive layout algorithms, but would widen drawings in our setting; reflection symmetry can be added as a width-preserving post-processing step~\cite{Ploeg14, Walker90}, so our hardness result and algorithms are unaffected.

We also do not impose the stronger \emph{subtree separation} property of some classical algorithms~\cite{ChanGKT02,GargR04,ReingoldT81}, which requires bounding rectangles of node-disjoint subtrees not to overlap, since it can substantially widen drawings and is incompatible with the compact tidy conventions we follow~\cite{MarriottS04,Ploeg14}.

\subsection{Problem statement}

For a sized tree $T$ with prescribed widths, every choice of sibling order $\pi$ induces a unique tidy drawing $\Gamma_\pi$ of minimum width subject to (A1)--(A3); the algorithm of van der Ploeg~\cite{Ploeg14}, an extension of the linear-time Reingold--Tilford procedure to sized vertices, computes $\Gamma_\pi$ in $O(n)$ time for a given order $\pi$.
Different choices of $\pi$ yield drawings of different widths, sometimes by a substantial factor (\cref{fig:teaser}).
We write $W^\pi(T) = W(\Gamma_\pi)$ for the width of the tidy drawing under sibling order $\pi$, and define
\begin{equation}
W^\ast(T) = \min_{\pi} W^\pi(T)
\label{eq:wstar}
\end{equation}
to be the minimum achievable width over all sibling orders.

\begin{definition}[\textsc{Min-Width Tree Drawing with Reordering}, \mwtdr]
\label{prob:mwtdr}
Given a sized rooted tree $T = (V, E, w, h)$, find a sibling order $\pi^\ast$ such that $W^{\pi^\ast}(T) = W^\ast(T)$.
\end{definition}

In its decision form, \mwtdr takes an additional input $W \in \mathbb{R}_{\geq 0}$ and asks whether $W^\ast(T) \leq W$.
We will primarily work with the optimization version, but our hardness result in \cref{sec:hardness} establishes hardness of the decision form by reduction from a known \textsf{NP}-complete decision problem.

\section{Hardness}
\label{sec:hardness}

\begin{figure}[t]
  \centering
  \begin{subfigure}[t]{0.42\textwidth}
    \centering
    \includegraphics[width=\textwidth,page=1]{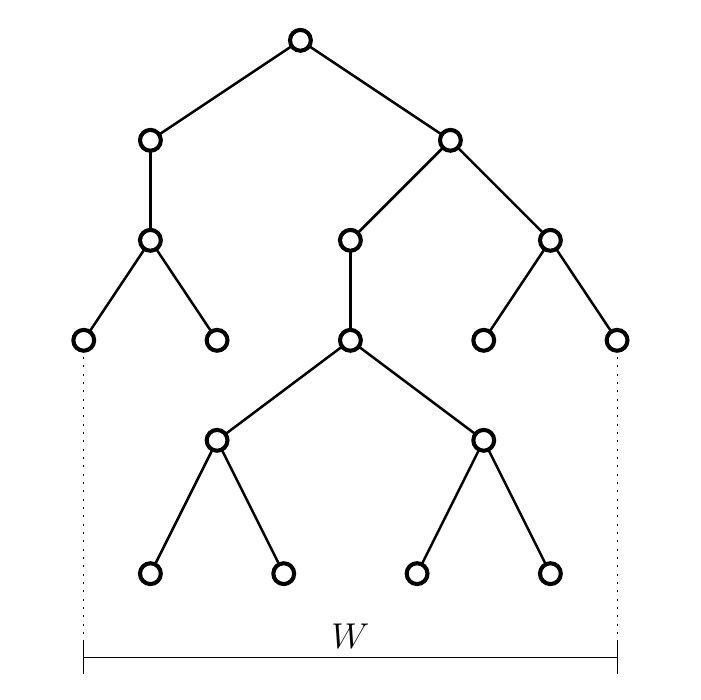}
  \end{subfigure}
  \hfill
  \begin{subfigure}[t]{0.42\textwidth}
    \centering
    \includegraphics[width=\textwidth,page=2]{figures/MS-reduction.pdf}
  \end{subfigure}
\caption{A \textsc{ULBTL} instance drawn with point vertices (left) and the corresponding \mwtdr instance with unit-width boxes (right). Identifying point coordinates with bounding-box centers preserves all aesthetics; the bounding-box width exceeds the point-coordinate range by exactly one. Gaps between boxes are for visual clarity only, and neighboring boxes touch in the drawing.}
  \label{fig:reduction}
\end{figure}

We show \mwtdr is \textsf{NP}-complete by reduction from \textsc{Unordered Layered Binary Tree Layout} (\textsc{ULBTL}), shown \textsf{NP}-complete by Marriott and Stuckey~\cite{MarriottS04}.
A \textsc{ULBTL} drawing places each vertex of a binary tree at integer-spaced layers, separates same-layer vertices by horizontal distance at least one, and centers each parent at the average of its children's $x$-coordinates; the sibling order is free. The width is $\max_v x_v - \min_v x_v$. Note that \textsc{ULBTL} does not restrict to integer positions for vertices, as previous results~\cite{SupowitR82} already show this is \textsf{NP}-hard for ordered trees.

\begin{theorem}\label{thm:hardness}
\mwtdr is \textsf{NP}-complete, even on binary trees with unit-width vertices.
\end{theorem}

\begin{proof}
The problem is in \textsf{NP}, as we can verify in polynomial time whether a given tree has width less than or equal to the provided value, given the sibling order as a certificate. Given the sibling order $\pi$, the corresponding tidy drawing $\Gamma_\pi$ is computed in $O(n)$ time with  van der Ploeg's algorithm~\cite{Ploeg14}, and the realized width is obtained as the difference between the rightmost and leftmost sides of the vertex bounding boxes. 

The hardness argument is via reduction from \textsc{ULBTL}. Given a \textsc{ULBTL} instance $T$, set $w(v) = 1$ for every $v \in V$ and inter-layer distance $\delta = 1$ to obtain an \mwtdr instance $T'$ on the same underlying tree.
We claim $T$ has a \textsc{ULBTL} drawing of width $W$ if and only if $T'$ has a tidy layered drawing of width $W + 1$.
Identifying the point coordinate $x_v$ in $T$ with the bounding-box center $x(v) = x_v$ in $T'$: the layered constraint matches the vertical rule of \cref{sec:background} since both place children one unit below their parents; non-overlap (A1) matches the unit horizontal separation between same-layer points; sibling order (A2) is preserved by construction; and centering (A3) for binary nodes coincides with the \textsc{ULBTL} averaging rule directly, regardless of vertex widths.
The bounding-box width exceeds the point-coordinate range by exactly one, since each box extends a half-unit beyond the point on each side.
The reverse direction is symmetric. \cref{fig:reduction} illustrates the reduction.
\end{proof}

\section{An Exact Mixed Integer Linear Program}
\label{sec:milp}

We formulate \mwtdr as a mixed-integer linear program (MILP) that simultaneously determines the sibling order at every internal vertex and the corresponding horizontal coordinates.
The optimum equals $W^\ast(T)$ as defined in~\eqref{eq:wstar}.

\subsection{Formulation}

The MILP has three groups of variables: a continuous coordinate $x(v) \in \mathbb{R}_{\geq 0}$ for every $v \in V$ giving the horizontal center of $v$, a continuous variable $W \in \mathbb{R}_{\geq 0}$ for the layout width, and a binary variable $o_{uv} \in \{0,1\}$ for every ordered pair $(u,v)$ of siblings with the convention that $o_{uv} = 1$ iff $x(u) > x(v)$, together with the identification $o_{vu} \equiv 1 - o_{uv}$.
For cross-parent same-depth pairs $(u, v)$, the order is inherited from the ancestors: we write $o_{uv}$ as a shorthand for $o_{\mathrm{parent}(u),\, \mathrm{parent}(v)}$, applied recursively until the parents are siblings.
This avoids introducing redundant binary variables.
We use a single big-$M$ constant $M = \sum_{v \in V} w(v)$, which upper-bounds the layout width: by non-overlap, any layer's extent is at most the sum of its vertex widths, which is in turn at most $M$.

The objective function of the MILP minimizes the layout width:
\begin{equation}
\min\ W
\label{eq:objective}
\end{equation}
The remaining constraints encode the aesthetic criteria (A1)--(A3) from \cref{sec:background} and the relationship between $W$ and the horizontal coordinates.

\textbf{Non-overlap (A1).}
For every pair $(u, v) \in L_i$ at depth $i$, the centers must be separated
by at least the sum of half-widths:
\begin{align}
x(u) - x(v) &\geq \tfrac{w(u) + w(v)}{2} - M (1 - o_{uv}), \label{eq:noover-1}\\
x(v) - x(u) &\geq \tfrac{w(u) + w(v)}{2} - M\, o_{uv}. \label{eq:noover-2}
\end{align}
Exactly one of the two inequalities is active depending on $o_{uv}$; the other is rendered inactive by the big-$M$ slack.

\textbf{Layer ordering (A2).}
For every internal vertex $v$ and every triple of distinct children $u, v', z \in \mathrm{children}(v)$, transitivity
\begin{equation}
o_{u v'} + o_{v' z} - o_{u z} \leq 1
\label{eq:trans}
\end{equation}
forbids the cyclic assignment $o_{u v'} = o_{v' z} = 1$ with $o_{u z} = 0$.
Cross-parent same-depth pairs inherit their ancestors' order via the aliasing convention, which combined with non-overlap prevents subtrees of different parents from interleaving.

\textbf{Parent centering (A3).}
For every internal vertex $v$, the centering rule of (A3) places $x(v)$ at the midpoint of the centers of its leftmost and rightmost children.
For binary trees, both children are unambiguously identifiable and the rule is a direct equality:
\begin{align}
x(v) &= \tfrac{1}{2}\bigl(x(c_1) + x(c_2)\bigr) && \text{if } \mathrm{children}(v) = \{c_1, c_2\}, \label{eq:center-binary}\\
x(v) &= x(c) && \text{if } \mathrm{children}(v) = \{c\}. \label{eq:center-unary}
\end{align}
For internal vertices with three or more children, the leftmost and rightmost children depend on the sibling order, which is itself a decision variable.
We introduce two auxiliary continuous variables $x_L(v), x_R(v) \in \mathbb{R}$ representing the centers of the leftmost and rightmost children of $v$.
Writing $k = |\mathrm{children}(v)|$ and, for each $c \in \mathrm{children}(v)$, $r_c = \sum_{c' \neq c} o_{c', c}$ for the number of siblings to the right of $c$ and $\ell_c = \sum_{c' \neq c} o_{c, c'}$ for the number to the left, we constrain
\begin{align}
x_L(v) &\leq x(c) \leq x_R(v), \label{eq:xLR-bounds}\\
x_L(v) &\geq x(c) - M(k - 1 - r_c), \qquad x_R(v) \leq x(c) + M(k - 1 - \ell_c), \label{eq:xLR-tight}
\end{align}
for every $c \in \mathrm{children}(v)$.
Inequality~\eqref{eq:xLR-bounds} makes $x_L(v)$ a lower bound and $x_R(v)$ an upper bound on the children's centers.
The Big-M inequalities in~\eqref{eq:xLR-tight} tighten these bounds at exactly the leftmost and rightmost children: when $c$ is the leftmost child, all $k-1$ other siblings sit to its right, so $r_c = k-1$ and the Big-M slack vanishes, forcing $x_L(v) = x(c)$; symmetrically for $x_R(v)$ at the rightmost child.
The centering equation is then
\begin{equation}
x(v) = \tfrac{1}{2}\bigl(x_L(v) + x_R(v)\bigr).
\label{eq:center-nary}
\end{equation}

\textbf{Width and anchoring.}
The objective variable $W$ must dominate the right edge of every vertex's bounding box, and the layout must be anchored to a fixed origin so that $W$ has a well-defined meaning.
We enforce both via two inequalities, one per vertex:
\begin{equation}
W \geq x(v) + \tfrac{w(v)}{2}, \qquad x(v) \geq \tfrac{w(v)}{2} \qquad \forall\, v \in V.
\label{eq:width-bound}
\end{equation}
The first inequality is the standard formulation of $W = \max_v(x(v) + w(v)/2)$.
Since the objective minimizes $W$, this inequality is tight for at least one vertex (the rightmost), making $W$ exactly the width of the layout as defined in~\eqref{eq:width}.
The second inequality fixes the layout's translational degree of freedom: without it, any feasible solution can be shifted horizontally by an arbitrary amount.
Requiring $x(v) \geq w(v)/2$ for every $v$ bounds every left edge below by $0$, and minimization of $W$ then drives at least one left edge to exactly $0$, so $W$ equals the absolute width of the drawing.

\textbf{Model size.}
The dominant constraint groups are the non-overlap inequalities, which contribute $O(n^2)$ constraints in the worst case, and the transitivity inequalities, which contribute $O(n \cdot \Delta^2)$ from the same-parent scope with $\Delta$ being the max degree.

\textbf{Correctness.}
Any feasible $(x, o, W)$ corresponds to a unique sibling order $\pi$ and a tidy layered drawing with $W$ equal to its width by~\eqref{eq:width-bound}.
Conversely, every $\pi$ and its tidy drawing yield a feasible $(x, o, W)$.
Minimizing $W$ over feasible solutions therefore returns $W^\ast(T)$.

\subsection{Implementation Details}

The formulation above admits several refinements in our implementation.

\textbf{Symmetry breaking.}
Every feasible drawing has a mirror image of equal width obtained by reversing the sibling order at every internal vertex.
We break this global reflection symmetry by fixing the relative order of two children of the root: if the root has at least two children, we choose a canonical pair $c_1, c_2 \in \mathrm{children}(r)$ and add the bound
\begin{equation}
o_{c_1, c_2} = 1.
\label{eq:sym-break}
\end{equation}
This eliminates one representative from each mirror pair, halving the search tree without affecting the optimum. 
If the root has fewer than two children, we apply the bound at the first descendant with at least two children.

\textbf{Warm start.}
We seed the solver with a feasible solution obtained by laying out $T$ in the input sibling order using a post-order traversal: leaves are placed left-to-right at consecutive horizontal positions, and internal vertices are placed at the midpoint of their leftmost and rightmost children.
The resulting positions and ordering bits are passed to Gurobi as an MIP start, providing an immediate upper bound on $W^\ast(T)$ that the solver uses to prune the branch-and-bound tree.

During development we observed that symmetry breaking noticeably reduced
solve times, while the warm start had only a minor effect, presumably
because Gurobi quickly finds feasible solutions of comparable quality on
its own.
We did not perform a systematic ablation of these refinements and leave a
quantitative assessment to future work.

\section{A Simulated Annealing Heuristic}
\label{sec:heuristic}

The MILP scales poorly beyond a hundred vertices.
For larger instances, we use simulated annealing~\cite{KirkpatrickGV83,LaarhovenA87} (SA) to search over sibling orders, evaluating each candidate by running van der Ploeg's algorithm~\cite{Ploeg14}, which yields a tidy layered drawing of minimum width in $O(n)$ time.
Several properties of \mwtdr make SA a natural fit.
The search space is purely combinatorial (a permutation per internal vertex), a local move (swapping two siblings) is cheap to generate, and the objective is fast to evaluate exactly via~\cite{Ploeg14}.
SA is a well-established metaheuristic for such combinatorial layout problems, also within graph drawing~\cite{DavidsonH96,NollenburgW23,NollenburgW24}, and requires only a handful of parameters.
We do not claim that SA is the best possible metaheuristic for \mwtdr, and a systematic comparison with alternatives, such as dedicated algorithms or other metaheuristics, is left for future work.

\textbf{State representation and neighborhood.}
The state $s$ of the search is a sibling order at every internal vertex.
A neighboring state is obtained by exchanging two siblings of an internal vertex of degree at least two.
This is illustrated in \cref{fig:sa_step}.
The initial state is obtained by assigning a random permutation of the children at every internal vertex.

\textbf{Acceptance and cooling.}
At each step, a neighboring state $s'$ is generated from the current state $s$ and accepted with the standard Metropolis criterion: with probability one if $E(s') \leq E(s)$, and with probability $\exp(-(E(s') - E(s))/\tau)$ otherwise, where $\tau$ is the current temperature.
The temperature follows a geometric schedule, multiplied by a cooling rate $\alpha \in (0, 1)$ at every iteration.
The algorithm terminates when the temperature falls below a fixed threshold $\tau_{min}$ or when a hard iteration cap is reached.

\textbf{Two-phase variant.}
To escape local minima, we use a two-phase variant.
The first phase runs to completion as above, then the second phase resumes from the best state of the first phase, with the temperature reset to a value proportional to the achieved width.

\textbf{Lazy evaluation.}
To reduce the per-iteration cost of energy evaluation, we evaluate the neighboring state's energy only every $k$ iterations.
On the intermediate iterations, the neighbor is accepted without evaluation, which trades solution quality for speed. 
The intermediate states are tentative, and if the candidate state at the next evaluation is rejected, all of them are discarded, and the search resumes from the last accepted state.

\textbf{Incremental evaluation.}
Each energy evaluation currently reruns van der Ploeg's
algorithm~\cite{Ploeg14} from scratch. An incremental variant appears
feasible as the algorithm computes a contour for every subtree in a bottom-up
pass, and a swap of two siblings at an internal vertex $v$ leaves all
contours below $v$ unchanged, so re-merging the children's contours at $v$
and along the path from $v$ to the root suffices to obtain the new width.
The worst-case cost per update remains linear, but the practical cost depends on the depth of the swap. 
We have not implemented this optimization and leave it as an engineering improvement for interactive settings.

\textbf{Parameters.}
We tuned the SA parameters using the hyperparameter-optimization framework Optuna~\cite{AkibaSYOK19}, minimizing the achieved width across 100 synthetic instances generated by the same generator as the evaluation benchmark.
The resulting configuration estimates $\tau_0$ per instance such that the initial acceptance probability is approximately $0.9$, and sets $\tau_{\min} = 0.01$, cooling rate $\alpha = 0.999$, a hard cap of $5000$ iterations, a lazy-evaluation interval $k = 5$, and a two-phase scaling factor $\beta = 2.0$.
Across the Optuna trials, solution quality was most sensitive to the cooling rate and largely robust to the remaining parameters.

\begin{figure}[t]
  \centering
  \begin{subfigure}[b]{0.45\textwidth}
    \centering
    \includegraphics[width=\textwidth]{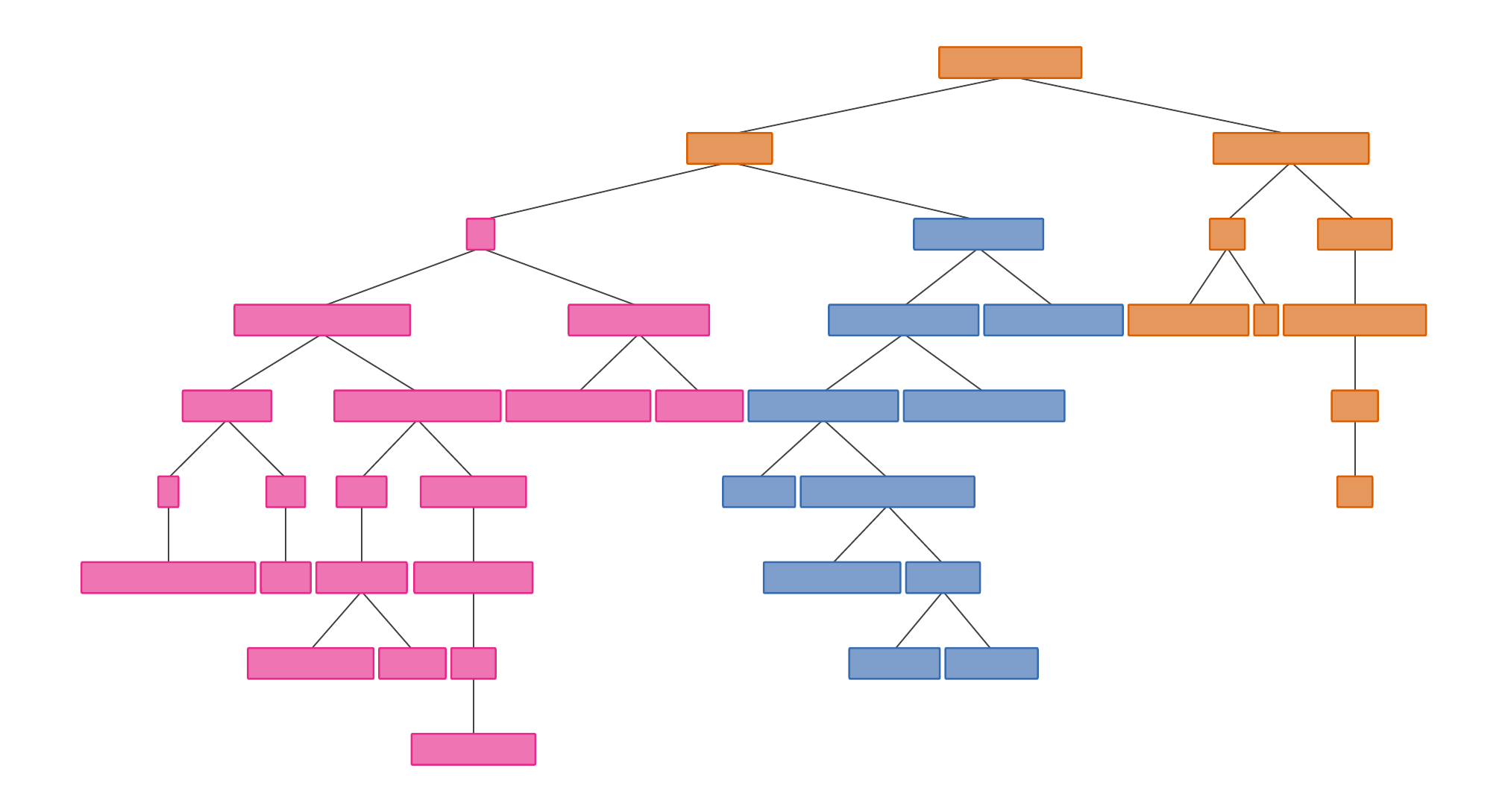}
  \end{subfigure}
  \hfill
  \begin{subfigure}[b]{0.45\textwidth}
    \centering
    \includegraphics[width=\textwidth]{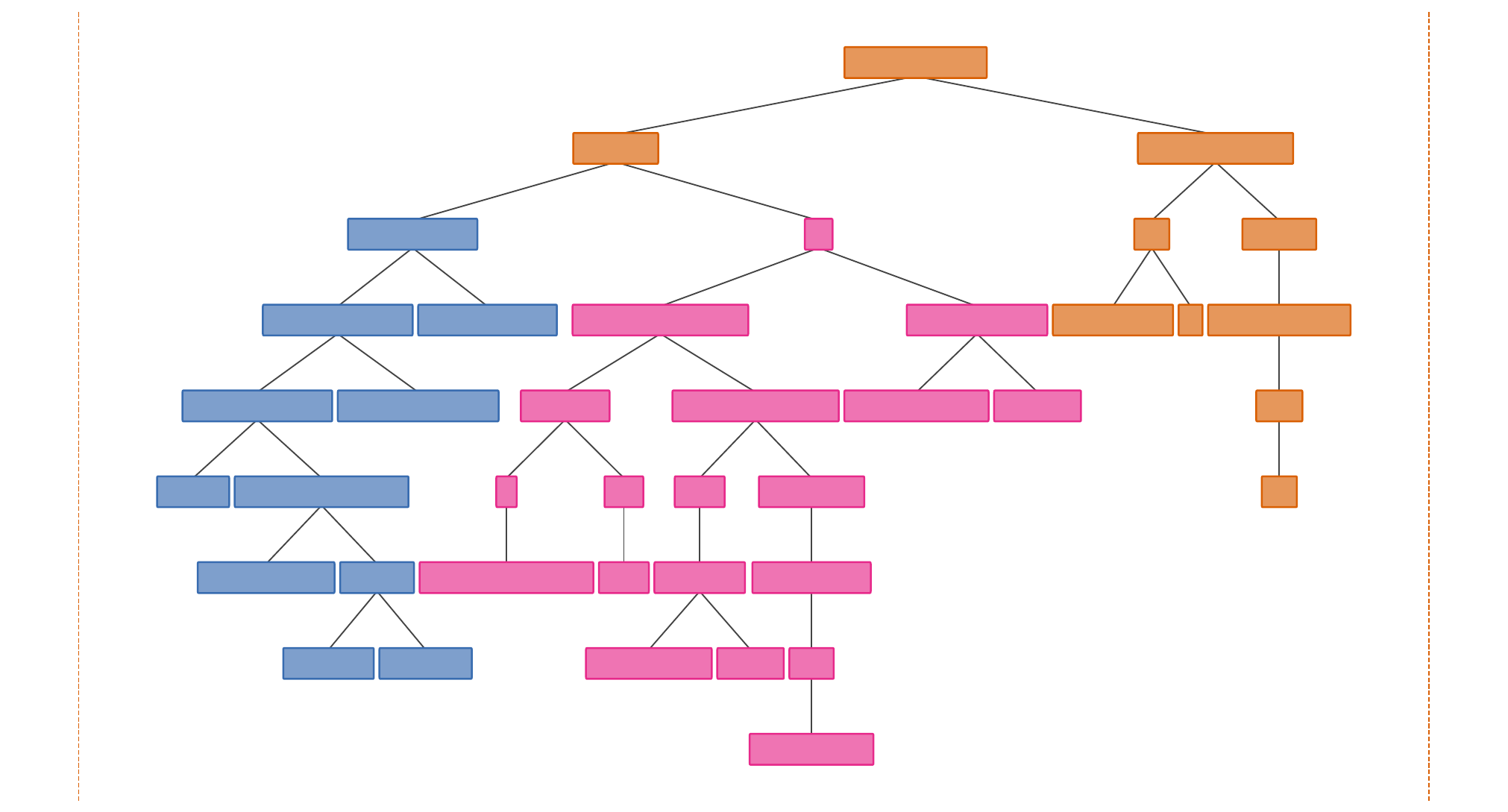}
  \end{subfigure}
  \caption{One step of the simulated annealing heuristic showing before (left) and after (right). The pink and blue vertices are swapped in the order of their parent, decreasing $W(T)$ by $12.6\%$.}
  \label{fig:sa_step}
\end{figure}
\section{Experimental Setup}
\label{sec:experiments}

We evaluate the MILP of \cref{sec:milp} and the simulated annealing heuristic of \cref{sec:heuristic} against a simple baseline on a suite of synthetic and real-world trees.
The evaluation has three goals.
First, to verify that the heuristic finds layouts close to the MILP optimum on instances where the MILP terminates.
Second, to characterize how each method scales as the input grows.
Third, to understand how tree shape and vertex-size variability affect the relative performance of the methods.
Source code, datasets, and scripts to reproduce all reported results are available in the supplemental material.

\subsection{Methods}

We compare three methods:
\begin{itemize}
  \item \mRandom. A single sibling order sampled uniformly at random from the search space, laid out with van der Ploeg's algorithm~\cite{Ploeg14}. This baseline measures whether non-trivial search is needed.
  
  \item \mSA. The simulated annealing heuristic of \cref{sec:heuristic}. 
  
  \item \mMILP. The mixed-integer linear program of \cref{sec:milp}, solved with Gurobi\footnote{https://www.gurobi.com/} with a time limit of 3600 seconds per instance. 
\end{itemize}

\subsection{Datasets}

\begin{figure}[t]
  \centering
  \begin{minipage}{0.05\textwidth}\strut\end{minipage}\hfill
  \begin{minipage}{0.30\textwidth}\centering Max degree 2\end{minipage}\hfill
  \begin{minipage}{0.30\textwidth}\centering Max degree 7\end{minipage}\hfill
  \begin{minipage}{0.30\textwidth}\centering Max degree 12\end{minipage}\\[2pt]
  \begin{minipage}{0.05\textwidth}\centering\rotatebox{90}{Shallow}\end{minipage}\hfill
  \begin{minipage}{0.30\textwidth}\includegraphics[width=\textwidth]{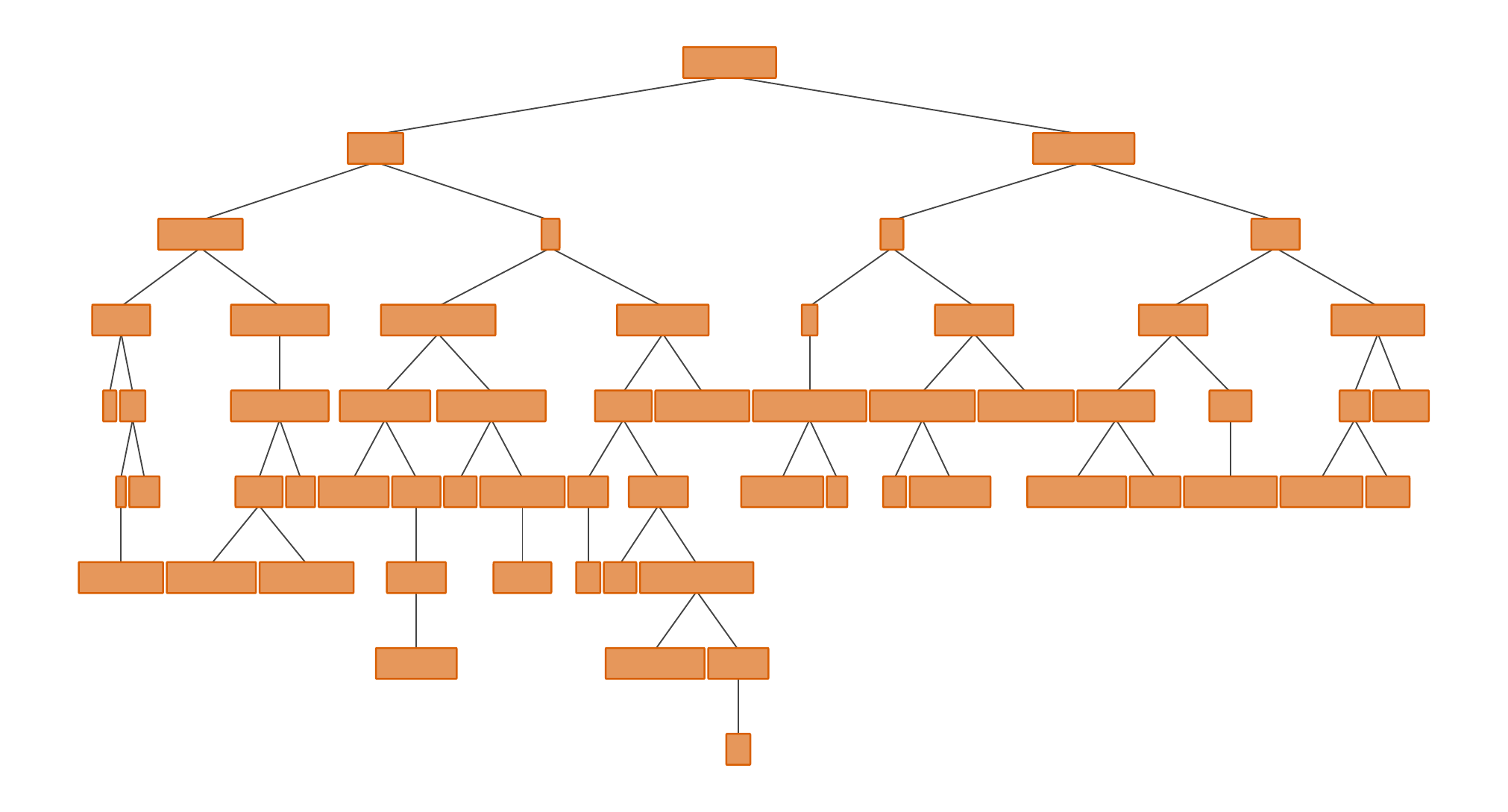}\end{minipage}\hfill
  \begin{minipage}{0.30\textwidth}\includegraphics[width=\textwidth]{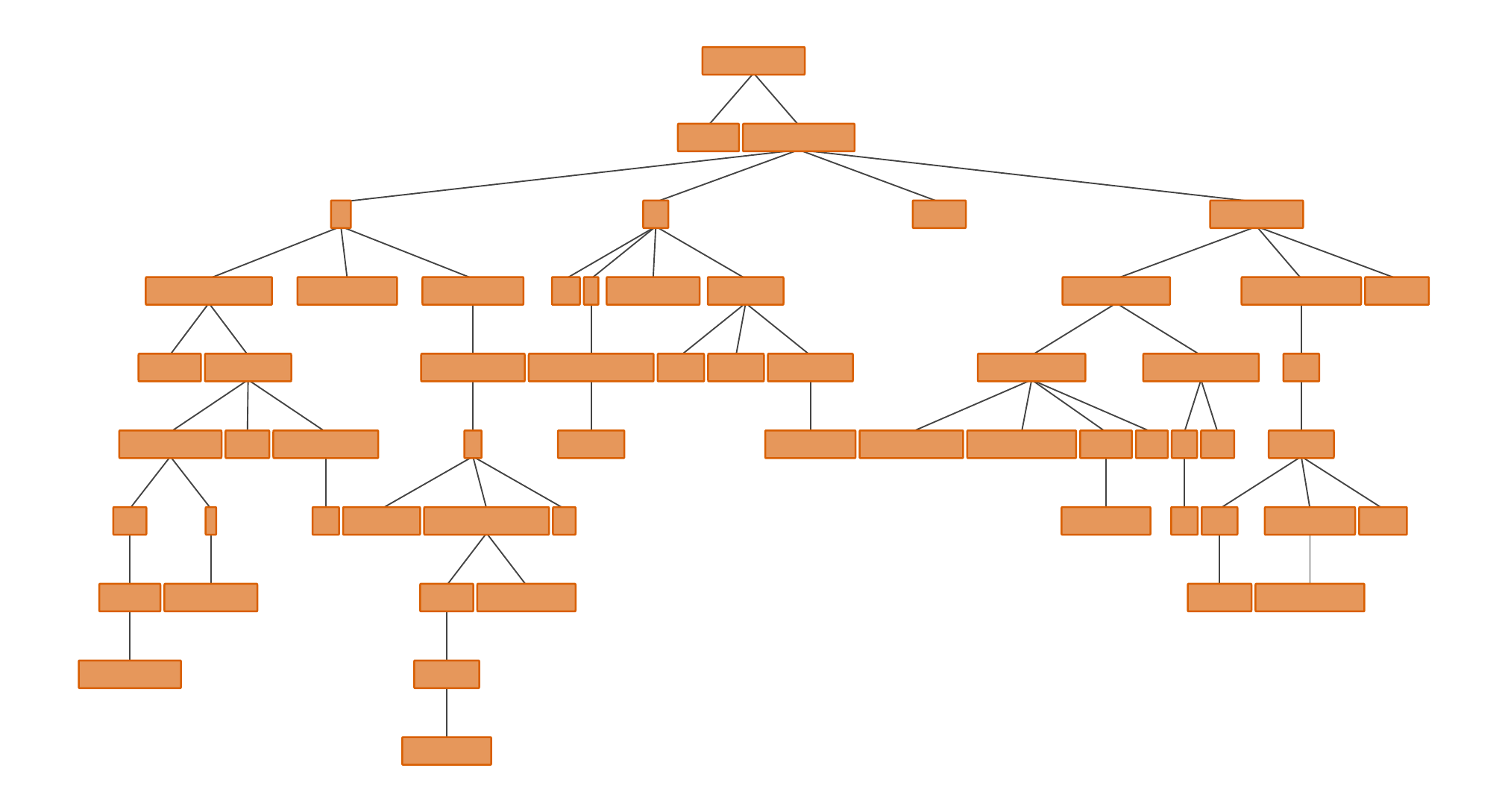}\end{minipage}\hfill
  \begin{minipage}{0.30\textwidth}\includegraphics[width=\textwidth]{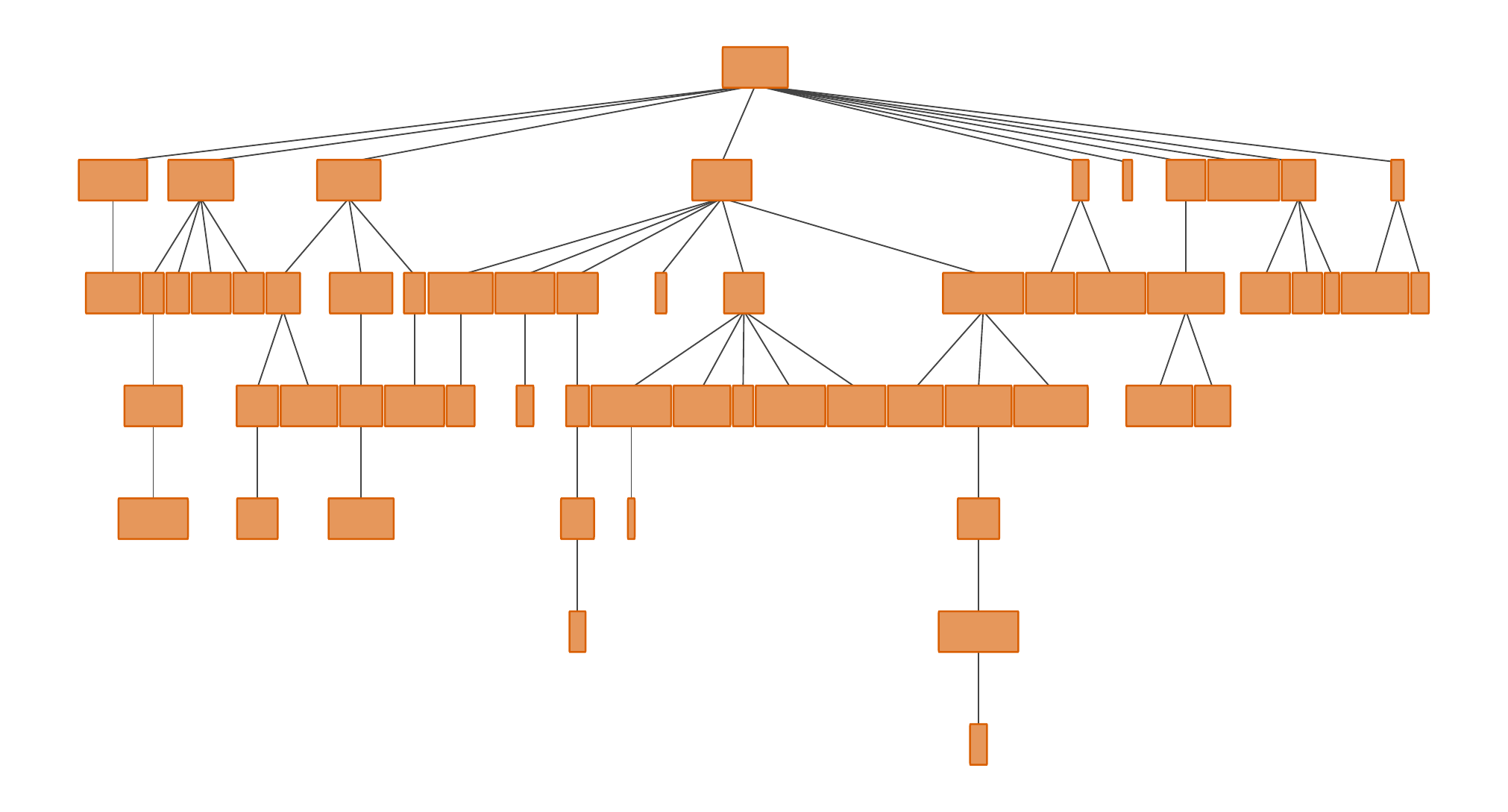}\end{minipage}\\[4pt]
  \begin{minipage}{0.05\textwidth}\centering\rotatebox{90}{Random}\end{minipage}\hfill
  \begin{minipage}{0.30\textwidth}\includegraphics[width=\textwidth]{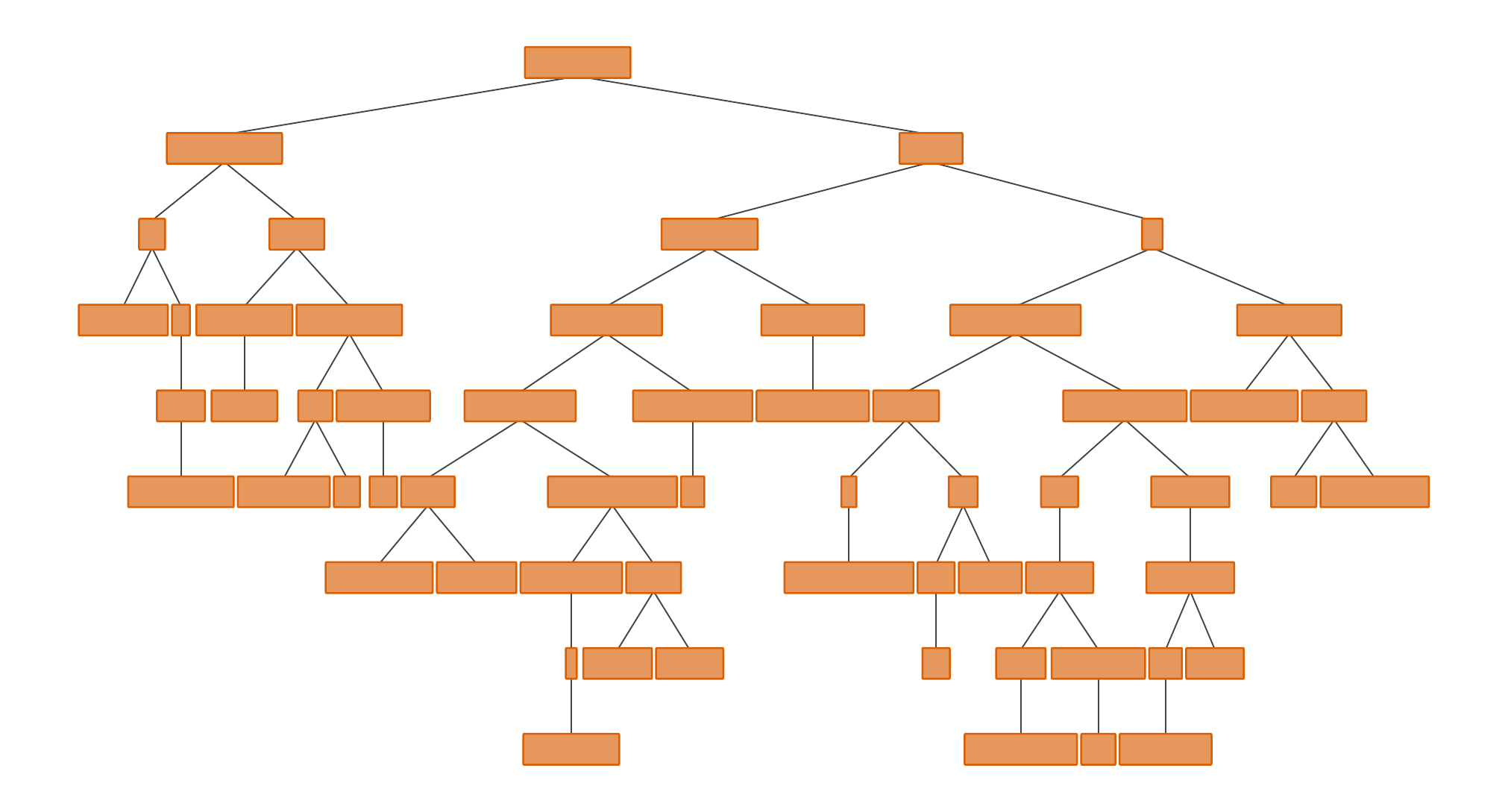}\end{minipage}\hfill
  \begin{minipage}{0.30\textwidth}\includegraphics[width=\textwidth]{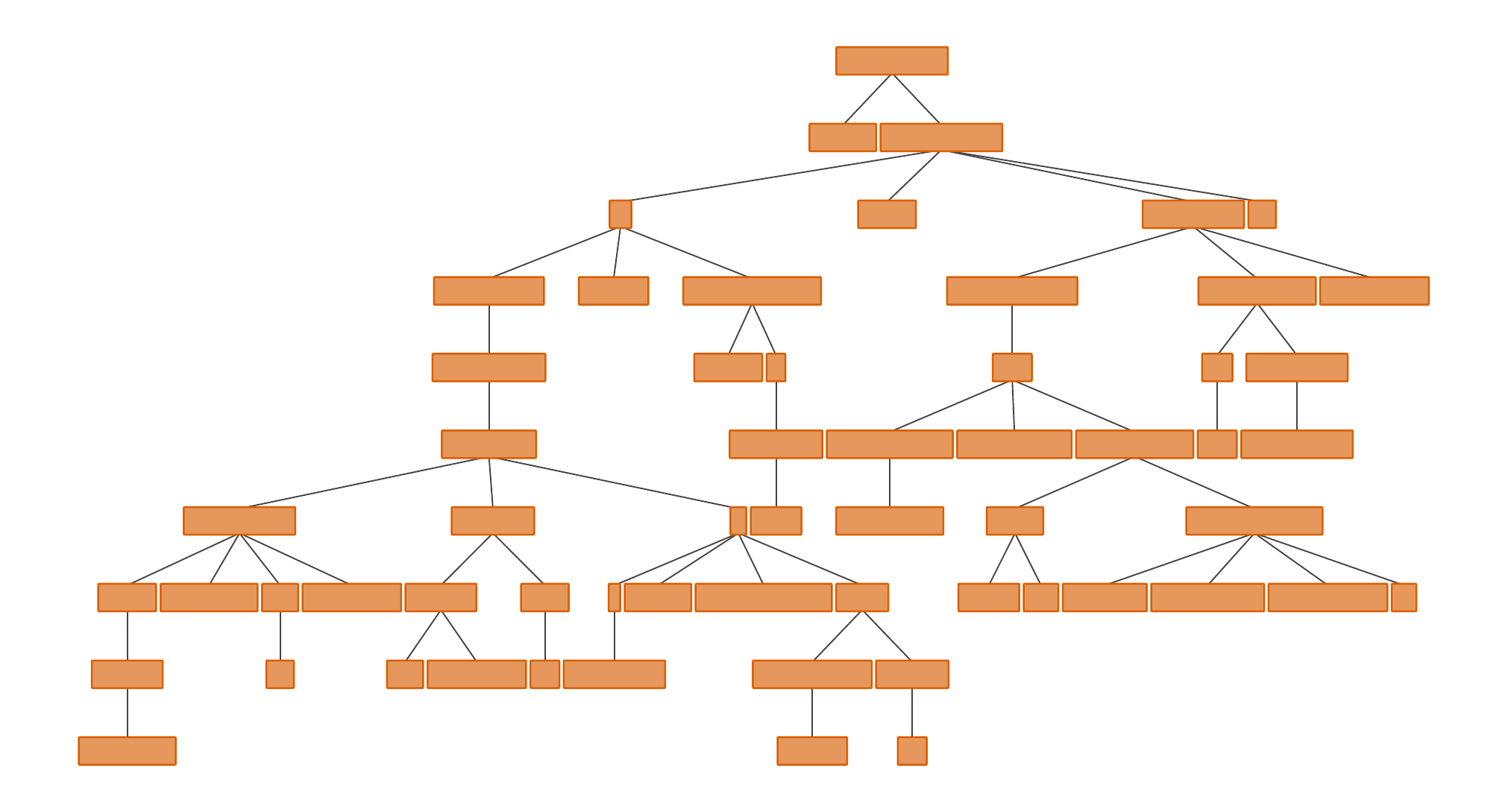}\end{minipage}\hfill
  \begin{minipage}{0.30\textwidth}\includegraphics[width=\textwidth]{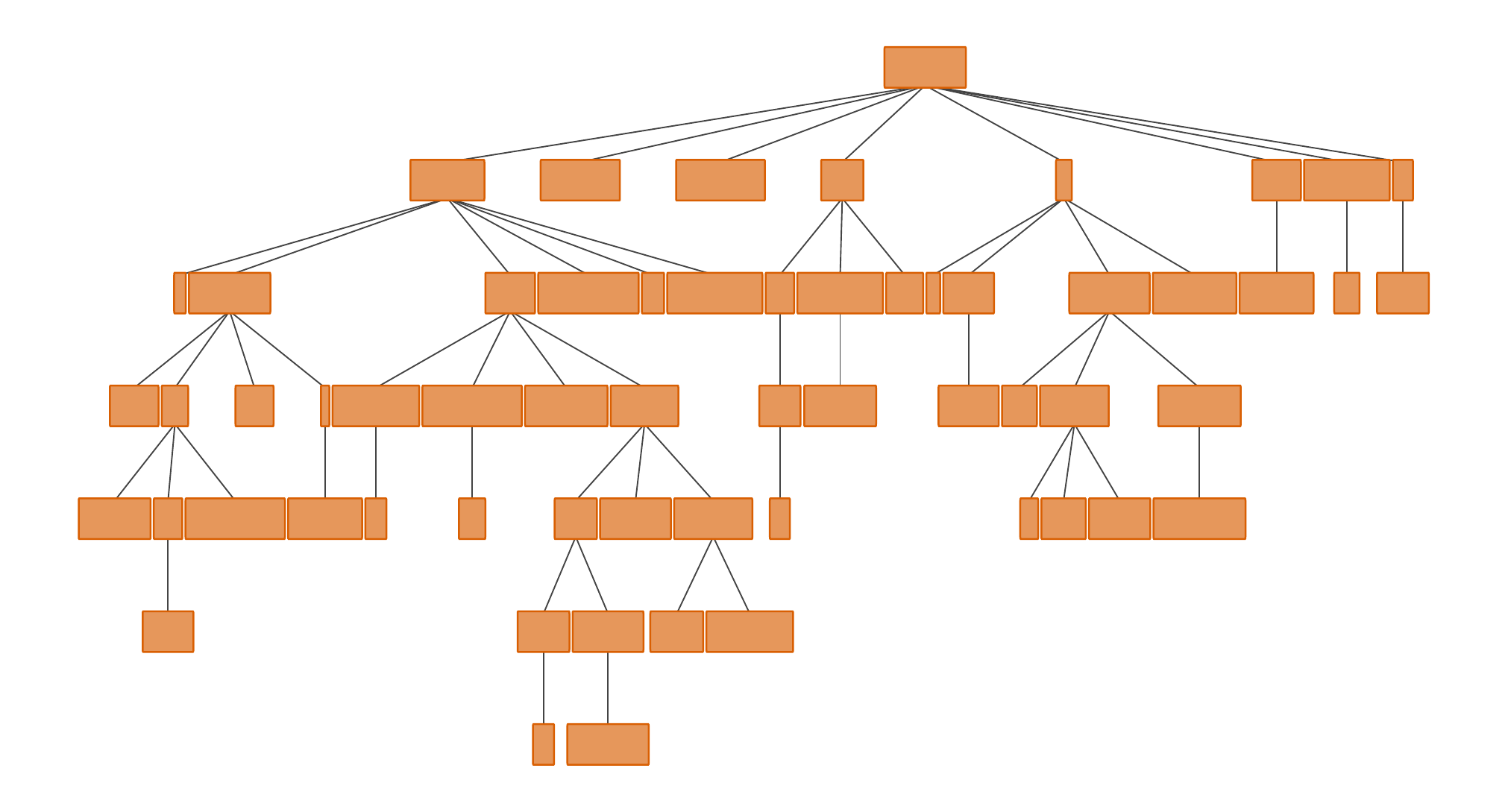}\end{minipage}\\[4pt]
  \begin{minipage}{0.05\textwidth}\centering\rotatebox{90}{Deep}\end{minipage}\hfill
  \begin{minipage}{0.30\textwidth}\includegraphics[width=\textwidth]{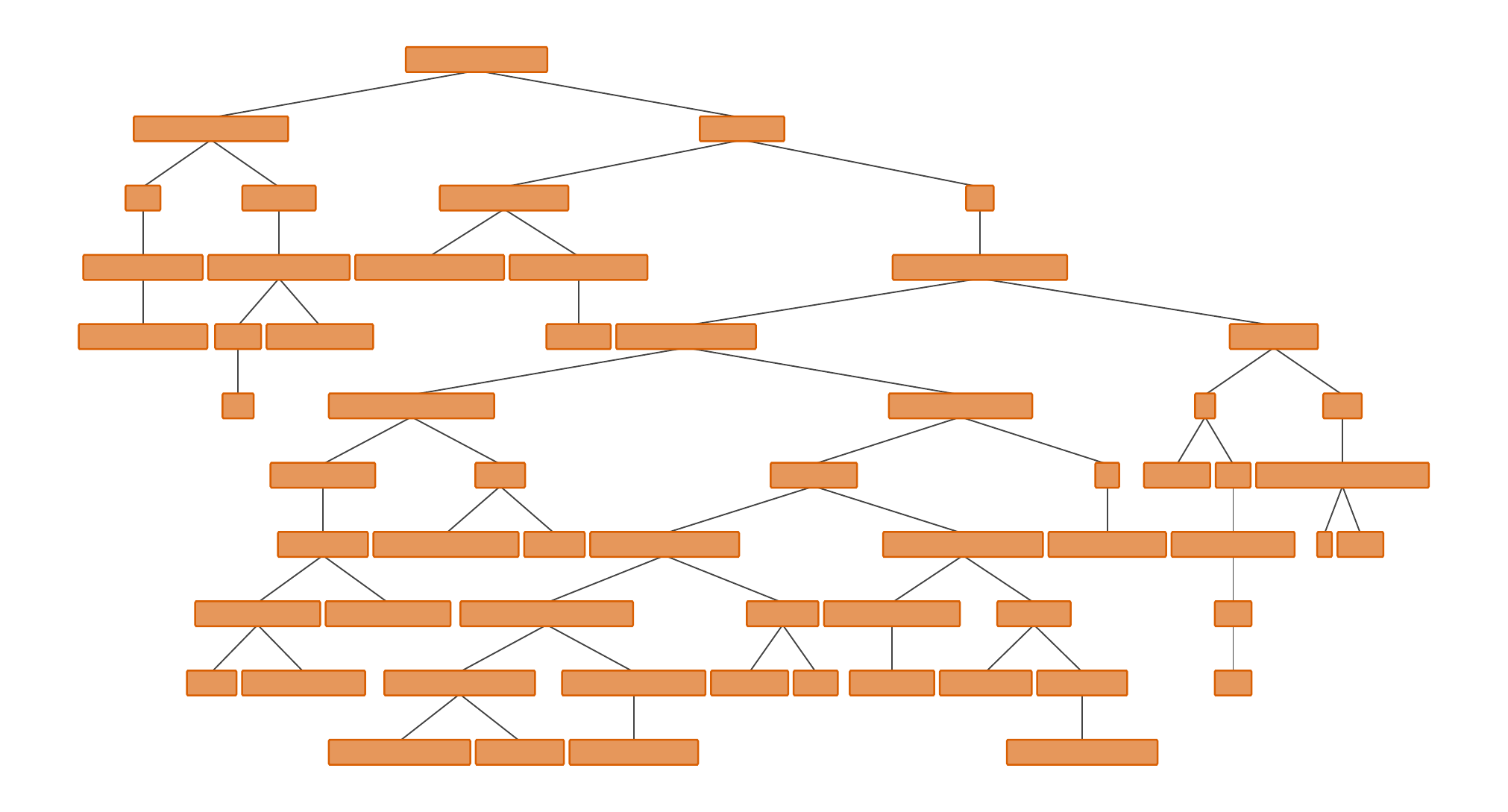}\end{minipage}\hfill
  \begin{minipage}{0.30\textwidth}\includegraphics[width=\textwidth]{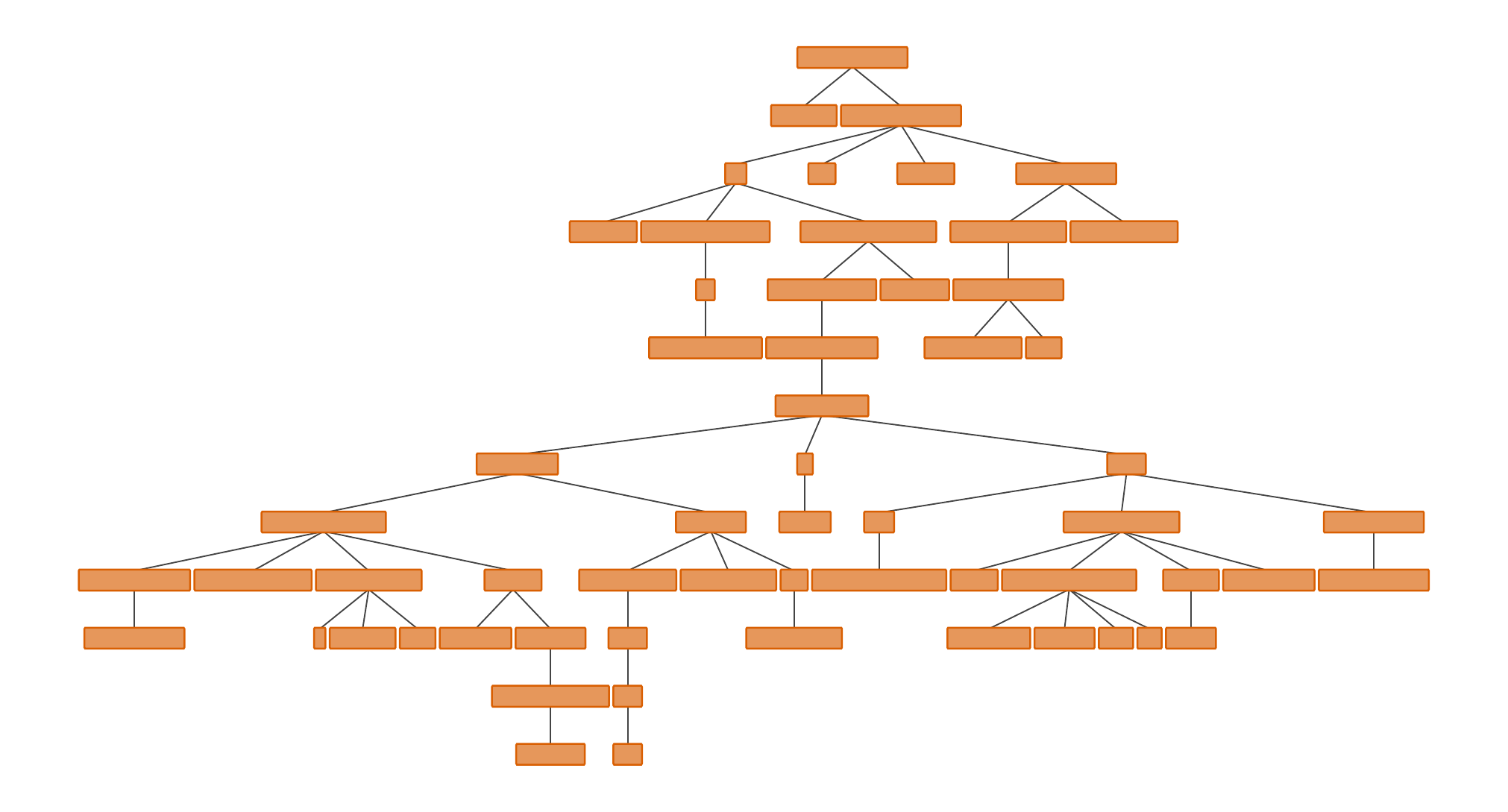}\end{minipage}\hfill
  \begin{minipage}{0.30\textwidth}\includegraphics[width=\textwidth]{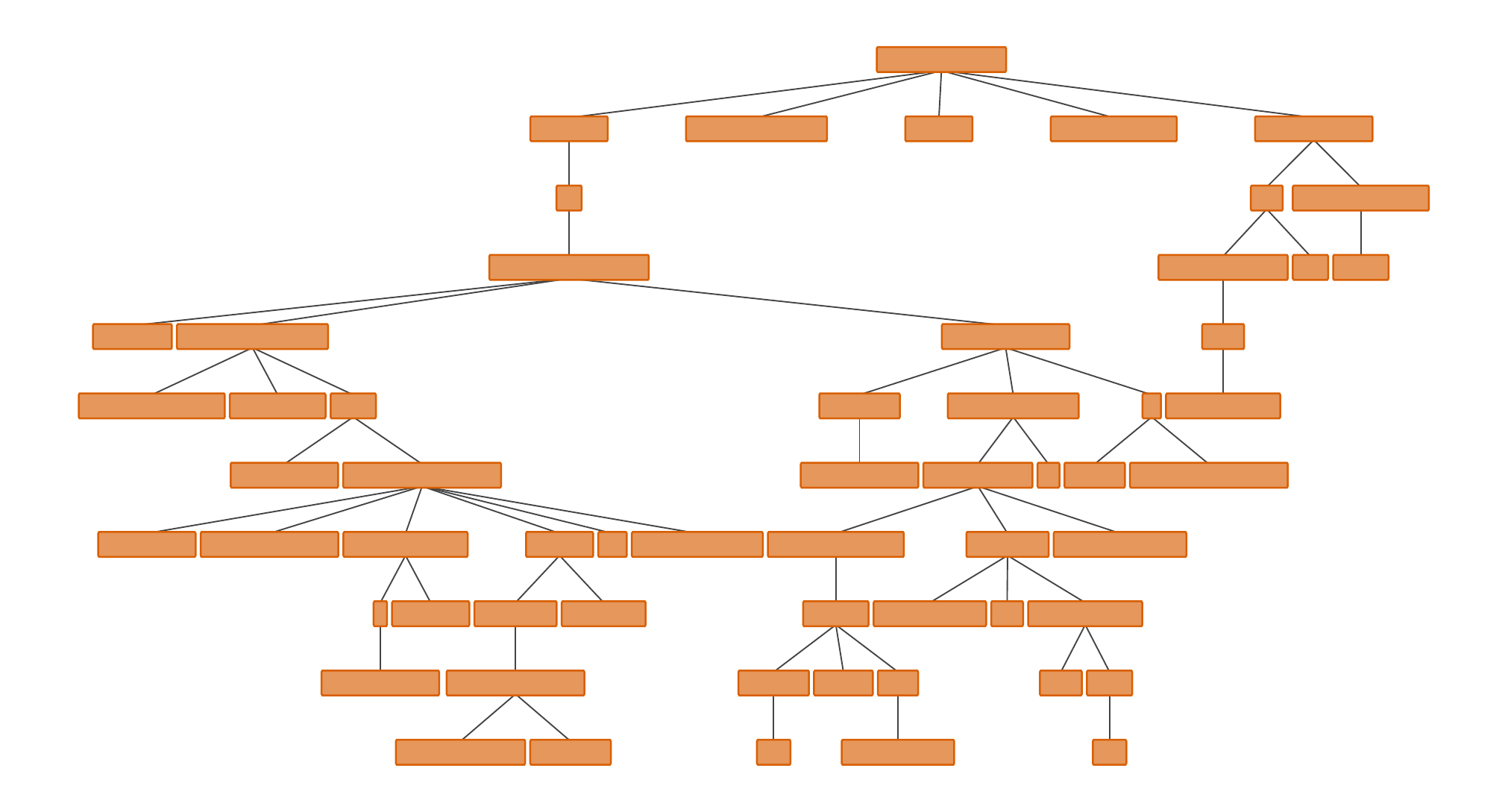}\end{minipage}
  \caption{Representative synthetic trees across the nine (shape, degree cap) combinations, each generated with $n = 60$ vertices. All instances are optimized with SA.}
  \label{fig:synthetic-grid}
\end{figure}

We test on synthetic trees generated by a parameterized random model, and on real-world hierarchies drawn from automated theorem-prover solutions.

\textbf{Synthetic trees.}
We generate trees of varying sizes and shapes using a depth-biased preferential attachment model.
The first vertex is the root.
Every other vertex $u$ is attached to an existing vertex $v$ chosen with probability proportional to $(1 + d(v))^\gamma$, where $d(v)$ is the depth of $v$ and $\gamma$ controls the shape:
\begin{itemize}
  \item \textbf{Shallow} ($\gamma = -2$): shallow vertices preferred; trees approach the shape of complete $\Delta$-ary trees.
  \item \textbf{Random} ($\gamma = 0$): every existing vertex equally likely.
  \item \textbf{Deep} ($\gamma = 2$): deep vertices preferred; trees feature long primary paths and few short branches.
\end{itemize}
We also cap the maximum degree at $2$ (binary trees), $7$, and $12$, by rejecting attachment proposals to vertices that have already reached the cap.
For every vertex, the width $w(v)$ is drawn independently from a uniform distribution on the interval $[1, 10]$.
We generate trees with $n \in \{20, 25, \dots, 200\}$ for each (shape, degree) with $5$ independent trees per combination (1665 total). We illustrate this for $n=60$ in \cref{fig:synthetic-grid}.

\textbf{Real-world data.}
We use proof trees from the TSTP archive~\cite{Sutcliffe17}, a repository of automated theorem-prover solutions.
A theorem proof is represented as a directed acyclic graph, where each vertex is a derivation step (an axiom or an inference rule application), and the character length of the associated logical formula determines the respective width.
We convert each DAG into a tree by duplicating any subderivation that is referenced more than once, so that every internal vertex has a unique parent.
In total, we sample 20 proof trees stratified by size (small/medium/large) with $n \in [13,867]$ vertices (TSTP dataset). An example is shown in \cref{fig:realworld1}.

\begin{figure}[t]
  \centering
  \begin{subfigure}[b]{0.48\textwidth}
    \centering
    \includegraphics[width=\textwidth]{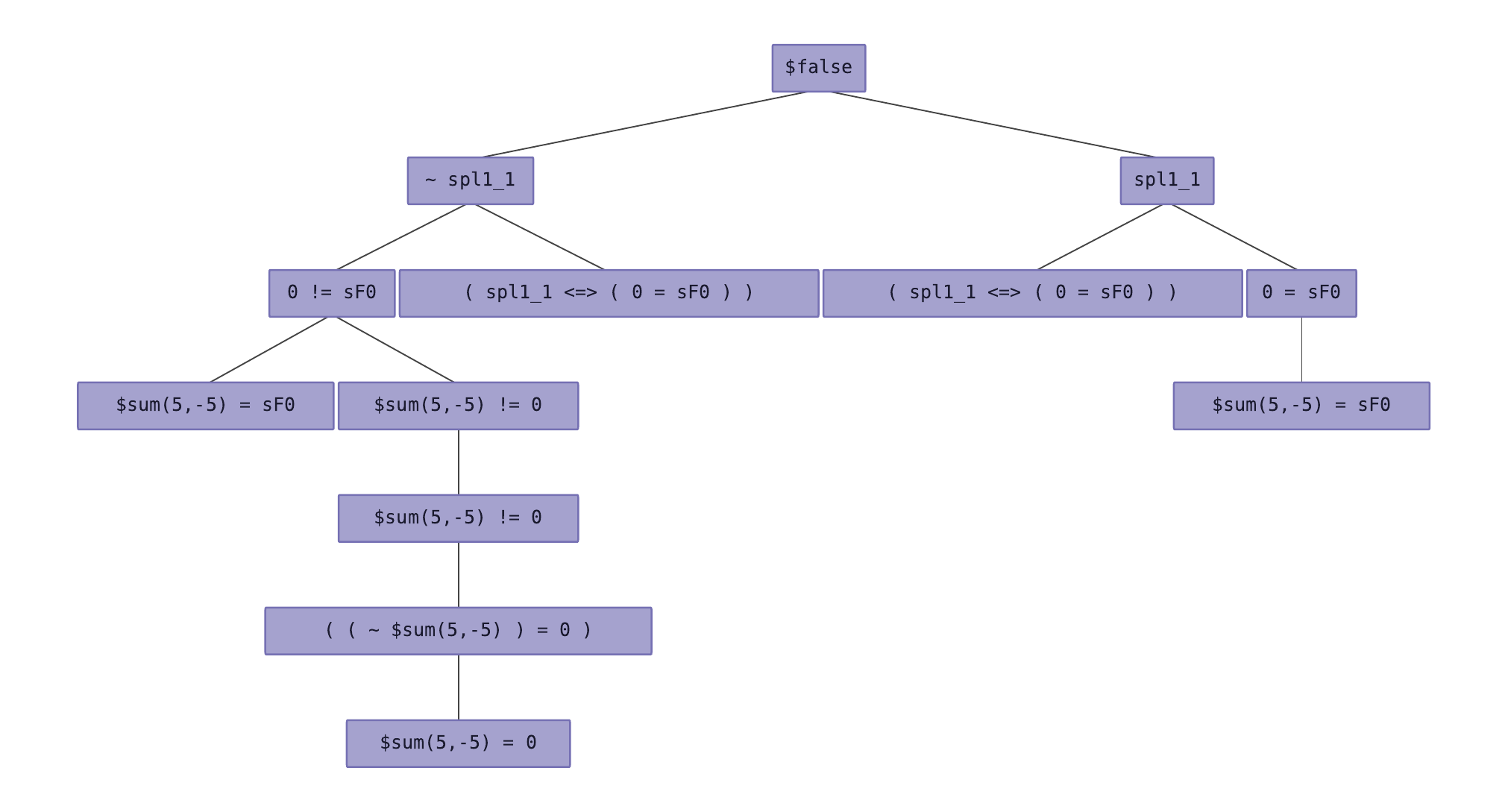}
  \end{subfigure}
  \hfill
  \begin{subfigure}[b]{0.48\textwidth}
    \centering
    \includegraphics[width=\textwidth]{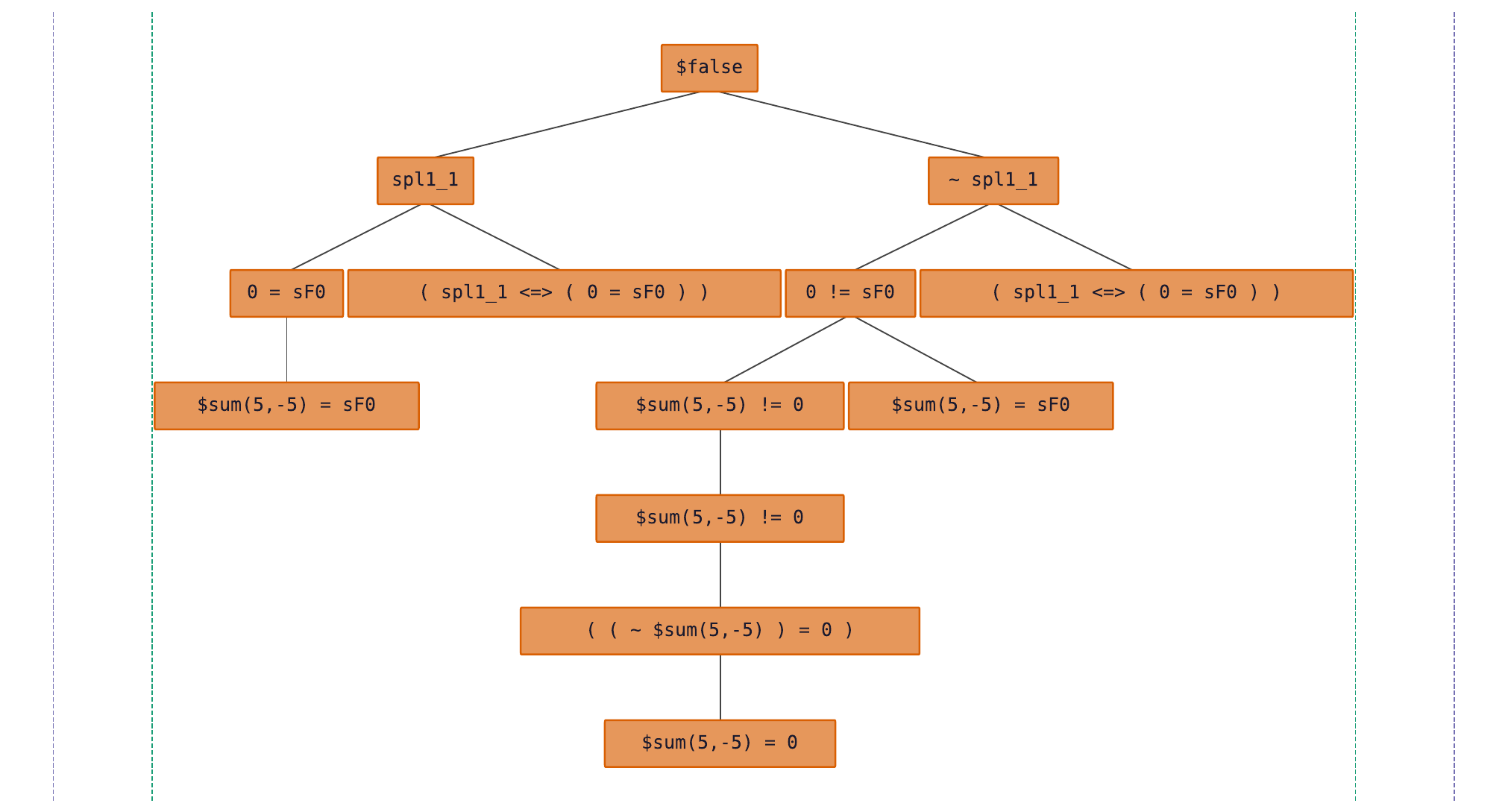}
  \end{subfigure}
  \caption{A real-world proof tree from the TSTP dataset (proof-3), drawn under (left) a random sibling order and (right) the order found by our SA heuristic, which achieves optimal quality. Reordering reduces $W(T)$ by $14.1\%$.}
  \label{fig:realworld1}
\end{figure}

\subsection{Setup and Metrics}
All experiments are run on a server with two Intel Xeon Platinum 8480+ CPUs, restricted to 4 cores and 32GB of RAM, running Ubuntu 24.04, with Python 3.12.
The MILP is solved using Gurobi 12.1 with a 3600 second time limit.
We reimplemented van der Ploeg's algorithm from an existing Rust implementation\footnote{https://github.com/zxch3n/tidy} in Python.

For each method $A$ on instance $T$, we record the achieved width $W_A(T)$ and the wall-clock runtime $\tau_A(T)$.
We compute the improvement $1 - W_A(T)/W_{\text{Random}}(T)$ over the Random baseline, and on instances where the MILP terminates, the optimality gap $W_A(T)/W^\ast(T)$.

\section{Experimental Results}
\label{sec:results}

We organize the results around four research questions on the synthetic data, followed by a separate evaluation on the real-world data to assess generalization.
Throughout, we summarize distributions of runtimes, widths, and improvements by their median together with the interquartile range (IQR), i.e., the interval $[Q_1, Q_3]$ between the $25^{\text{th}}$ and $75^{\text{th}}$ percentiles.
We report this rather than mean and standard deviation since several of our distributions are right-skewed (notably MILP runtime, which spans several orders of magnitude).

\subsection{Effect of Reordering on Drawing Width}
\label{sec:results:improvement}

\begin{figure}[t]
  \centering
  \includegraphics[width=\textwidth]{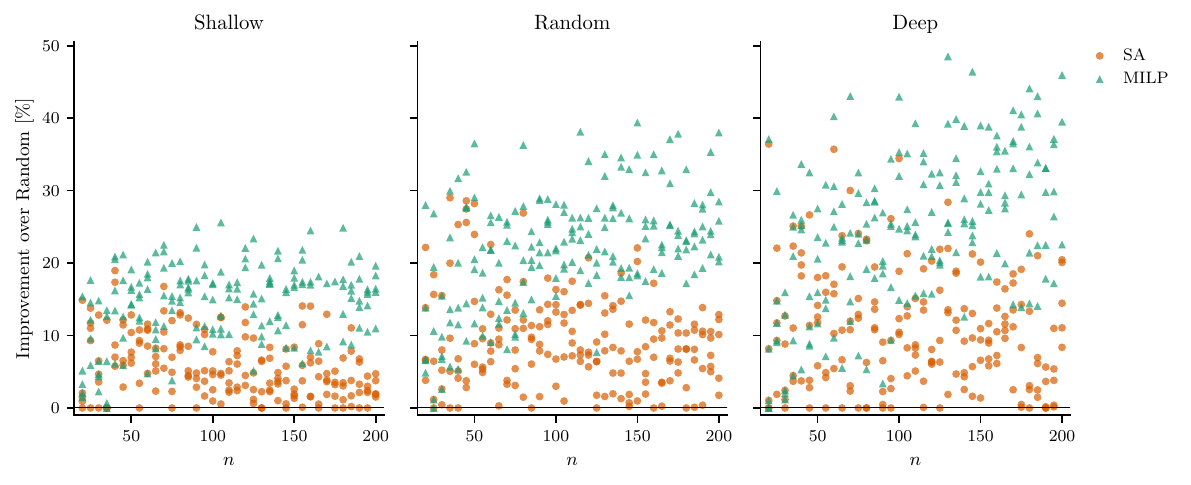}
  \caption{Improvement of SA and MILP over the random baseline as a function of tree size $n$, broken out by shape class.}
  \label{fig:results-improvement}
\end{figure}

We first ask whether reordering reduces drawing width by a meaningful amount, and how the improvement depends on tree shape (\cref{fig:results-improvement}).
On binary trees (degree cap $2$), the MILP reduces width over a random sibling order by a median of $20.1\%$ (IQR $[14.6\%, 26.2\%]$) across all instances, rising to $23.0\%$ (IQR $[16.7\%, 29.1\%]$) on trees with degree cap $12$.
The improvement depends substantially on tree shape: on binary deep trees the MILP achieves a median improvement of $25.3\%$, while on binary shallow trees only $15.6\%$, consistent with the geometric intuition that narrow, deep trees have more reordering room.
SA achieves a more modest median improvement of roughly $4\%$--$10\%$ across all configurations, with the same shape ordering: deep trees benefit most, shallow trees least.

\subsection{Quality of the Simulated Annealing Heuristic}
\label{sec:results:quality}
\begin{figure}[t]
  \centering
  \includegraphics[width=0.48\textwidth]{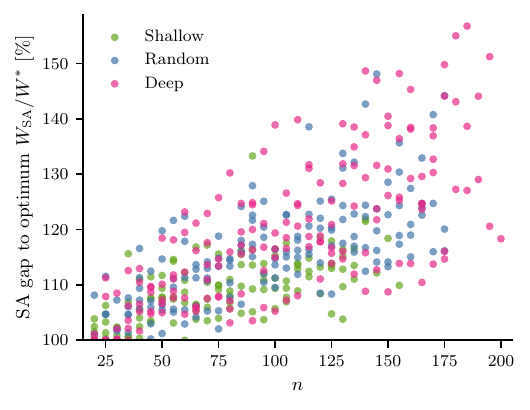}
  \caption{Relative gap $(W_{\text{SA}} - W^\ast)/W^\ast$ between SA and the MILP optimum on MILP-terminating instances, as a function of $n$.}
  \label{fig:results-gap}
\end{figure}

On instances where the MILP terminates within the time limit, we compare the width achieved by SA against the proven optimum (\cref{fig:results-gap}).
The median SA gap to the optimum is $13.6\%$ (IQR $[7.1\%, 20.6\%]$) on binary trees, $15.3\%$ on degree-$7$ trees, and $13.8\%$ on degree-$12$ trees, with the gap somewhat larger on deep trees than on shallow or random trees.
SA does not consistently reach the optimum on synthetic instances, but it is reliably within a $20-25\%$ relative gap for $75\%$ of instances and remains an order of magnitude faster than the MILP across the size range.

\subsection{Scaling Behavior}
\label{sec:results:scaling}

\begin{figure}[t]
  \centering
  \begin{minipage}{0.48\textwidth}
    \centering
    \includegraphics[width=\textwidth]{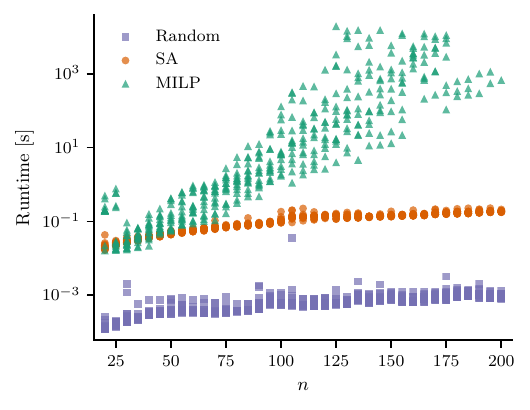}\\
  \end{minipage}\hfill
  \begin{minipage}{0.48\textwidth}
    \centering
    \includegraphics[width=\textwidth]{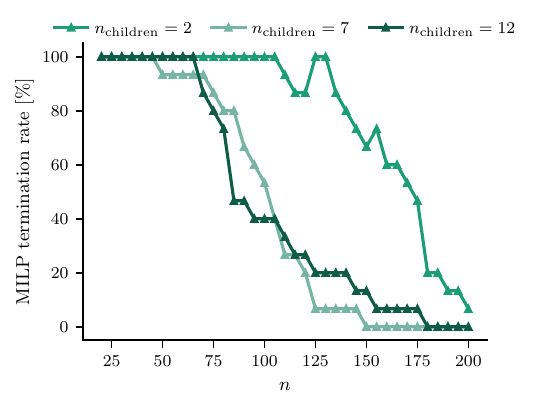}\\
  \end{minipage}
  \caption{Scaling behavior of the three methods.
  (left) Median wall-clock runtime per instance vs.\ $n$ with a log-scaled $y$-axis.
  (right) Fraction of MILP runs that terminated within the time limit.}
  \label{fig:results-scaling}
\end{figure}

We measure how the wall-clock runtime of each method grows with $n$ (\cref{fig:results-scaling}).
SA scales smoothly: median runtime is $0.13$--$0.17$\,s across all degree caps and tree sizes, and the IQRs remain narrow even at $n = 200$.
The MILP, in contrast, hits the $3600$-second time limit at sizes that depend strongly on the degree cap.
On binary trees, the MILP terminates on essentially every instance up to $n = 130$ and degrades gradually thereafter, dropping below $50\%$ termination only at $n = 175$.
On degree-$7$ trees, termination falls gradually from $n = 50$ onward, drops below $50\%$ at $n = 105$ ($40\%$), and below $10\%$ at $n = 125$; the MILP fails to solve any instance with $n \geq 150$.
On degree-$12$ trees, the cliff is even steeper, occurring at $n = 85$.
Overall, the MILP terminates on $80\%$ of binary instances but only $44\%$ and $45\%$ of degree-$7$ and degree-$12$ instances, respectively.
This behavior is explained by the structure of the search space rather than by model size, which grows only moderately (\cref{sec:milp}).
The number of sibling orders is the product of the factorials of the internal vertices' degrees, so a higher degree cap inflates the search space far more than additional vertices do.

\subsection{Effect of Tree Shape and Degree}
\label{sec:results:shape}

\begin{figure}[t]
  \centering
  \begin{minipage}{0.48\textwidth}
    \centering
    \includegraphics[width=\textwidth]{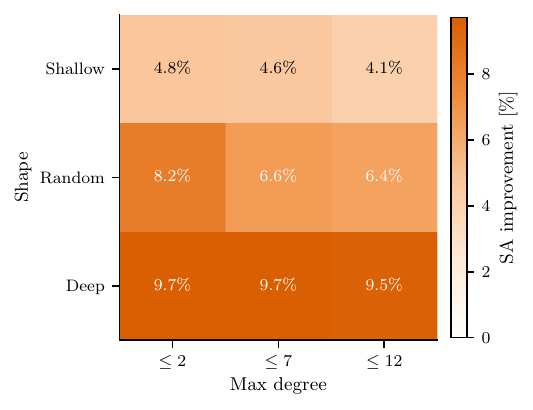}\\
  \end{minipage}\hfill
  \begin{minipage}{0.48\textwidth}
    \centering
    \includegraphics[width=\textwidth]{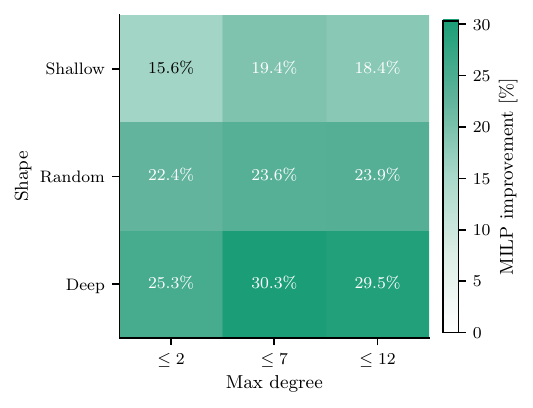}\\
  \end{minipage}
  \caption{Median improvement over the random baseline across the nine (shape, maximum degree) combinations, for (left) the SA heuristic and (right) the MILP optimum.}
  \label{fig:results-heatmap}
\end{figure}

We characterize when each method works best by comparing performance across the nine (shape, degree cap) combinations (\cref{fig:results-heatmap}).
SA improvement over Random varies from $4.1\%$ (shallow, degree $12$) to $9.7\%$ (deep, degrees $2$ and $7$); for a fixed shape the variation across degree caps is modest.
MILP improvement ranges from $15.6\%$ (shallow, degree $2$) to $30.4\%$ (deep, degree $7$).
For both methods, the shape effect dominates the degree effect: deep trees consistently outperform shallow trees regardless of degree cap, and varying the degree cap at a fixed shape changes the median improvement by about five percentage points.

\subsection{Real-world results}
\label{sec:results:realworld}

\begin{figure}[t]
  \centering
  \begin{subfigure}[b]{0.48\textwidth}
    \centering
    \includegraphics[width=\textwidth]{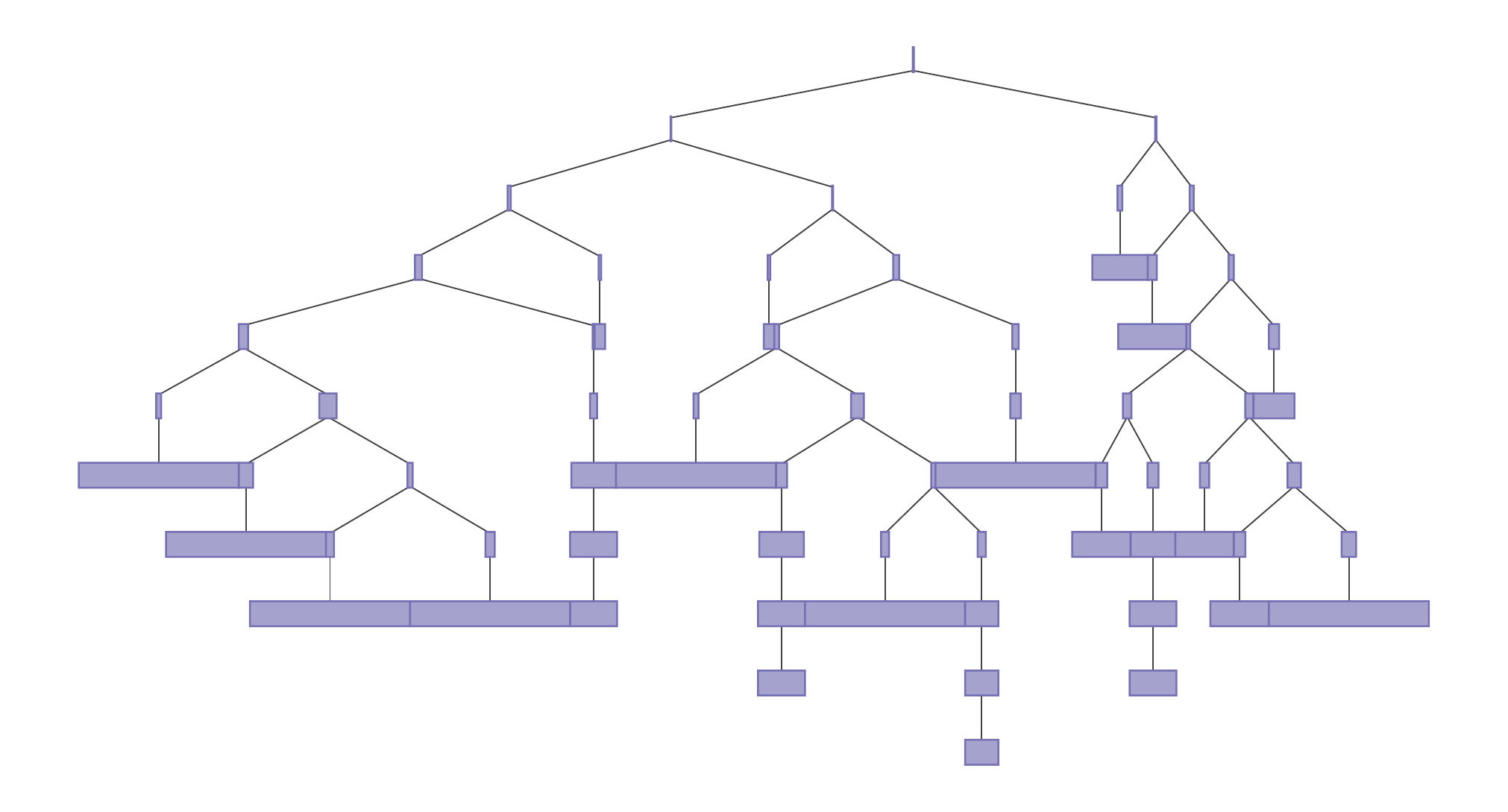}
  \end{subfigure}
  \hfill
  \begin{subfigure}[b]{0.48\textwidth}
    \centering
    \includegraphics[width=\textwidth]{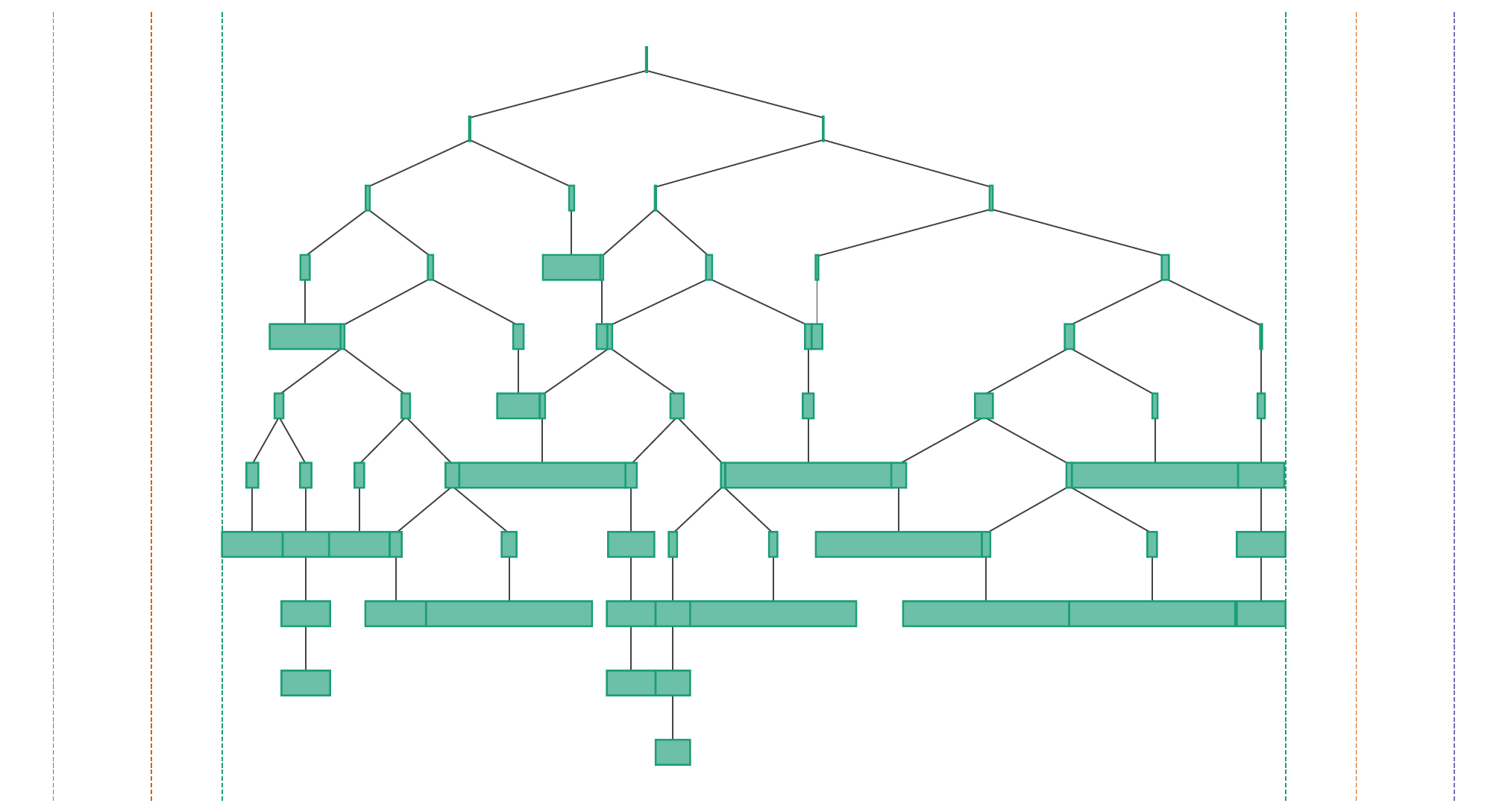}
  \end{subfigure}
  \caption{A real-world proof tree from the TSTP dataset (proof-19), drawn under (left) a random sibling order and (right) the order found by our MILP. Reordering reduces $W(T)$ by $24.1\%$.}
  \label{fig:realworld2}
\end{figure}

The MILP terminates within the time limit on 14 of the 20 instances; the remaining 6 instances ($n \in \{237, 310, 433, 867\}$ and two of $n \approx 70\text{--}80$) hit the 3600-second limit. The latter instances exhibit a very shallow structure.
On instances where the MILP terminates, the median improvement over random is $18.2\%$. \Cref{fig:realworld2} shows a representative case with $24.1\%$ reduction. 
SA matches or comes within 5\% of the MILP optimum on 11 of the 14 MILP-terminating instances, including the three cases where it achieves the optimum exactly.
On the largest instances (no MILP optimum available), SA improves over random by $0.7\%$ to $8.7\%$, a more modest range that reflects the structural property: deeper, sparser trees admit larger reordering gains than wide, flat trees.
Across the dataset, the heuristic produces a tighter layout than the random baseline on every instance except three of $n \leq 9$, where the search space is small enough that any sibling order is already near-optimal.

\section{Discussion}
\label{sec:discussion}

\textbf{Effect of tree shape on reordering room.}
Across both synthetic and real-world inputs, we observe a consistent pattern: the improvement from reordering depends strongly on the tree's shape, being largest for deep, sparse trees and smallest for wide, shallow ones.
The intuition is geometric.
At each layer, the horizontal extent is bounded below by the sum of the layer's vertex widths, and sibling ordering can only affect the width through the centering rule (A3): the parent's position depends on the leftmost and rightmost child, and propagates upward through ancestors.
When a layer is wide, for instance, because the tree has high fanout, or because vertex widths accumulate as the sum of subtree contents (as in filesystem hierarchies, document structures, or aggregate-size representations), the layer's extent is already close to the sum-of-widths lower bound, and reordering has little additional room to operate.
Conversely, when layers are narrow, the centering rule has substantial leverage: a deep subtree on one side can lean into space left by a shallow subtree on the other, and the resulting savings compound through several layers.
This observation matches our results, in which the deep class benefits substantially more from reordering than the shallow class, and our real-world experiments, in which hierarchies with high fanout or sum-of-children width semantics show modest improvements while sparser hierarchies show large ones.

\textbf{Synthetic benchmarks are harder than real-world instances.}
SA's gap to the MILP optimum differs sharply between the two datasets: $13$--$15\%$ on synthetic instances, but exact or within $5\%$ on eleven of the fourteen MILP-solved TSTP proof trees.
The depth-biased generator we use produces instances that stress the heuristic considerably more than the real-world hierarchies we tested.
We read this as evidence that the synthetic gap is a stress test rather than the heuristic's typical performance, and that practical deployments are likely to see results much closer to optimum.

\textbf{Complementary roles of the heuristic and the exact method.}
The two methods are best understood as complementary rather than competing.
SA delivers a layout in well under a second on every instance we tested, making it suitable for interactive tools.
The MILP, in contrast, is too slow for interactive use but reliably proves optimality on instances with up to roughly a hundred vertices and serves as a ground-truth oracle.
This second role is the more important contribution of the MILP: it allows future heuristics for \mwtdr to be benchmarked against a proven optimum on small-to-moderate instances, which is otherwise not possible for an \textsf{NP}-hard problem.

\textbf{Whitespace.}
Our objective minimizes drawing width unconditionally, but readability is not monotone in compactness.
Whitespace separates subtrees and guides the eye, and packing everything to the sum-of-widths lower bound can make adjacent subtrees visually indistinguishable.
\mwtdr should therefore be understood as one end of a spectrum: in practice,
a small whitespace budget or a sibling-proximity penalty likely
produces more readable drawings at marginal cost in total width.
Additionally, several aspects of readability are directly quantifiable, such as the
aspect ratio of the drawing, the uniformity of edge lengths, the angular
resolution at internal vertices, or the horizontal distance between
adjacent subtrees. 
However, it is currently unclear which metric is the appropriate choice from a user-centric perspective.

\textbf{Limitations.}
Our study has several limitations.
First, we compare only against a single heuristic and against a random baseline. 
A more competitive baseline, along with alternative metaheuristics or better-tuned annealing schedules, may close part of the gap between SA and the MILP optimum on synthetic data.
Second, our real-world dataset contains only twenty proof trees from a single domain, and broader conclusions about practical instances would require evaluation on hierarchies from other application areas.
Third, the MILP time limit of $3600$ seconds is a pragmatic rather than a principled choice: a longer limit would shift the termination cliffs in \cref{fig:results-scaling} to the right but would not change the overall scaling picture.
Fourth, we restrict to layered drawings with uniform vertex height, while
vertex widths may vary arbitrarily. Supporting non-uniform heights is a natural extension, but it changes the character of the problem: once boxes protrude into the space between layers, a drawing must avoid not only overlaps between bounding boxes but also overlaps between edges and boxes, so feasibility is no longer captured by the per-layer separation constraints (A1) alone. Additionally, a different objective function, such as the drawing area, must be considered instead.
We expect that both the MILP and SA heuristic, via additional edge--box separation constraints and an adapted objective function, can be adapted, but leave this to future work.
Fifth, our objective minimizes total width but does not distribute available space among children; classical algorithms~\cite{Ploeg14,Walker90} apply a post-processing pass that redistributes children evenly within the layer's contour, which our algorithms do not currently incorporate.

\section{Conclusion}
\label{sec:conclusion}

We introduced \mwtdr, the problem of choosing a sibling order at each internal node of a tree with sized vertices so as to minimize the width of the resulting layered tidy drawing.
We showed that the problem is \textsf{NP}-complete even for binary trees with unit-size vertices.
We then provided an exact MILP and a fast simulated annealing heuristic that performs well in practice.

There are several natural directions for future work. 
On the theoretical side, the approximability of \mwtdr is unknown: our reduction establishes only \textsf{NP}-hardness, and it would be interesting to know if a polynomial-time approximation exists.
On the modeling side, dropping the uniform-height assumption and removing the layered constraint to obtain the non-layered tidy convention of van der Ploeg~\cite{Ploeg14} would both broaden applicability, as would studying \mwtdr under richer drawing conventions such as hi-trees~\cite{MarriottSGPB11} and ortho-radial layouts.
On the evaluation side, our objective is a proxy for readability rather than a direct measure, and a user study comparing width-optimal drawings against drawings produced by alternative criteria such as area, aspect ratio, or total edge length would clarify which optimization target best serves human readers.

\bibliography{bibliography}

\appendix

\end{document}